\documentclass[submission, Phys]{SciPost}

\usepackage{amssymb,amsthm,bm}
\newcommand{\bZ}{\mathbb{Z}}
\newcommand{\bR}{\mathbb{R}}

\newtheorem{defn}{Definition}
\newtheorem{prop}{Proposition}
\newtheorem{thm}{Theorem}

\def\[#1\]{%
\begin{gather}
  #1
  \end{gather}
}

\hypersetup{
colorlinks = true,
linkcolor= blue,
citecolor=[rgb]{0.15,0.65,0.13},
urlcolor = [rgb]{0.25, 0.41, 0.88}}

\usepackage{multirow}
\usepackage[normalem]{ulem}
\usepackage{tikz}
\usepackage{tikz-cd}
\usetikzlibrary{decorations.pathmorphing}
\usetikzlibrary{automata,positioning}
\tikzset{snake it/.style={decorate, decoration=snake}}
\usepackage{mathtools}
\DeclarePairedDelimiter\bra{\langle}{\rvert}
\DeclarePairedDelimiter\ket{\lvert}{\rangle}
\DeclarePairedDelimiterX\braket[2]{\langle}{\rangle}{#1 \delimsize\vert #2}

\usepackage{soul}

\begin{document}

\begin{center}{\Large \textbf{
Higgs Condensates are\\
Symmetry-Protected Topological Phases:\\
\large II. $U(1)$ Gauge Theory and Superconductors
}}\end{center}

% TODO: write the author list here. Use initials + surname format.
% Separate subsequent authors by a comma, omit comma at the end of the list.
% Mark the corresponding author with a superscript *.
\begin{center}
Ryan Thorngren\textsuperscript{1,2},
Tibor Rakovszky\textsuperscript{3},
Ruben Verresen\textsuperscript{4},
and Ashvin Vishwanath\textsuperscript{4}

\end{center}

\begin{center}
{\bf 1} Kavli Institute for Theoretical Physics, University of California, Santa Barbara, CA 93106, USA
\\
{\bf 2} Institut des Hautes Études Scientifiques, Bures-sur-Yvette, 91440, France
\\
{\bf 3} Department of Physics, Stanford University, Stanford, CA 94305, USA
\\
{\bf 4} Department of Physics, Harvard University, Cambridge, MA 02138, USA

\end{center}

\begin{center}
\today
\end{center}

\section*{Abstract}
{\bf
Classifying Higgs phases within the landscape of gapped and symmetry preserving states of matter presents a conceptual challenge.
We argue that $U(1)$ Higgs phases are symmetry-protected topological (SPT) phases and we derive their topological response theory and boundary anomaly---applicable to superconductors treated with dynamical electromagnetic field.
This generalizes the discussion of discrete gauge theories by \emph{Verresen et al., arXiv:2211.01376}\cite{partI}.
We show that a Higgs phase in $d$ spatial dimensions is in a non-trivial SPT class protected by a global $U(1)$ symmetry associated with the Higgs field, and a $d-2$ form $U(1)$ magnetic symmetry, associated with the absence of magnetic monopoles. In $d=2$, this gives an SPT with a mixed Hall response between conventional symmetries, whereas in $d=3$ we obtain a novel SPT protected by a 0-form and 1-form symmetry whose 2+1d boundary anomaly is satisfied by a superfluid.
The signature properties of superconductors---Higgs phases for electromagnetism---can be reproduced from this SPT response.
For instance, the Josephson effect directly arises from the aforementioned boundary superfluid. In addition to this minimalist approach being complementary to Landau-Ginzburg theory, its non-perturbative nature is useful in situations where fluctuations are significant.
We substantiate this by predicting the stability of the Josephson effect upon introducing monopoles in $U(1)$ lattice gauge theory, where tuning from the charge-1 Higgs phase to the confined phase leads to a quantum critical point in the junction.
Furthermore, this perspective reveals unexpected connections, such as how persistent currents at the surface of a superconductor arise from generalized Thouless pumps.
We also treat generalizations to partial-Higgs phases, including ``2e" condensates in electronic superconductors,  corresponding to symmetry-enriched topological orders.
}

\vspace{10pt}
\noindent\rule{\textwidth}{1pt}
\tableofcontents\thispagestyle{fancy}
\noindent\rule{\textwidth}{1pt}
\vspace{10pt}

\section{Introduction}

The Anderson-Higgs mechanism is of central importance to gauge theories. In the standard model of particle physics it gives mass to particles by electroweak symmetry breaking. In condensed matter systems the Higgsing of the electromagnetic field leads to superconductivity. 
On the other hand, the completely Higgs'd phase, in which a fundamental Higgs field is condensed, appears trivial from an infrared (IR) point of view: it has a unique gapped ground state, and can even be deformed to a confined phase without a phase transition \cite{Fradkin79}. So why is it so special, supporting for instance persistent currents in a superconductor?

Of late, it has been appreciated that in the presence of a \emph{global} symmetry,  even a system with a unique gapped ground state may be in a non-trivial phase. These phases are the  symmetry-protected topological (SPT) phases \cite{AKLT88,Gu09,pollmann_entanglement_2010,Turner11class,Fidkowski11class,chen_complete_2011,Schuch11,Senthil15} which include topological insulators \cite{Kane05,Qi08,Hasan10,Moore10}. Their most important feature is the gapless edge modes which appear at their boundary.

In this paper and its prequel \cite{partI}, we explain how the Higgs phase may be understood as an SPT (or more generally, in the presence of topological order, a symmetry \emph{enriched} topological (SET) phase), and how several of its most interesting features can be phrased in topological terms. In part I we discussed discrete gauge theories, while in this part II we will discuss $U(1)$ gauge theories.

Our work is informed by recent developments in the theory of higher form global symmetries \cite{generalizedglobsym,highersymmetry,Benini_2019,Jian_2021}. These are generalizations of our notion of symmetry, whose quantum numbers are carried by strings or other extended objects, rather than point charges. Such symmetries naturally occur in gauge theories. For instance, in Maxwell electromagnetism, magnetic field lines form closed loops, which we can think of as a conservation law for magnetic flux through arbitrary surfaces. This conservation law is equivalent to the absence of magnetic monopoles. The symmetry associated with this conserved quantity is called the magnetic symmetry, and plays a central role in our discussion. Above we mentioned that the Higgs phase can often be deformed into a confined phase \cite{Fradkin79}. This is only possible by adding fluctuating monopoles or otherwise {\em breaking} the magnetic symmetry; such a \emph{symmetry protected} robustness is characteristic of an SPT.

We show quite generally that for gauge theories with a magnetic symmetry, the Higgs phase is an SPT labelled by the global charge of the Higgs condensate (with more words to say when the gauge field is not completely Higgs'd). In particular, such phases cannot be symmetrically deformed into one another without a phase transition. Furthermore, there are protected gapless modes at interfaces between Higgs phases with differently charged condensates.

Another aspect of the Higgs phase which resonates with the SPT viewpoint is that, although it looks like a symmetry breaking state in a particular choice of gauge, it is actually invariant under all global symmetries, and in particular it has no Goldstone bosons. The fact that it is instead an SPT recalls the ``hidden symmetry breaking" picture of SPT phases \cite{kennedy1992hidden,elsehidden}, where in particular in 1D, nonlocal string operators develop long range order.

We apply this reasoning to superconductors, which are Higgs phases of the electromagnetic $U(1)$ gauge field. We derive, to the best of our knowledge for the first time, the topological response and edge anomaly associated with a superconductor in $d$ space dimensions which incorporates, in general, both global and  higher form symmetry. Our approach thus builds on the recent improved understanding of higher-form symmetries.

We find that the SPT point of view gives new perspective on Josephson junctions \cite{Josephson62} and persistent currents. In particular, at a strongly insulating Josephson junction, across which tunnelling is completely forbidden, the relative charge between either side of the junction has a mutual anomaly with the magnetic symmetry, and so we predict a protected edge state.
On general anomaly-matching grounds, we argue that the relative charge symmetry is always spontaneously broken at the interface. We can use this anomaly perspective to derive the AC Josephson effect. The DC effect is recovered when we add weak tunnelling that \emph{explicitly} breaks the symmetry.

We also discuss hybrid Josephson junctions where a global symmetry distinguishes the two condensates. This may give a new route to detecting exotic superconductivity. That is, we would like to be able to observe the global charges of the Higgs condensate, but the condensate itself is not gauge invariant, so there is no local observable we can measure. However, at an interface, we can compare the condensates from the two sides with a gauge-invariant operator. If the unknown condensate is actually charged under some global symmetry, then by anomaly-matching, there must be low energy charged interface modes, which we can detect by current measurements.

As for the persistent currents in a superconducting ring, these arise at an interface to the Coulomb phase, which spontaneously breaks the magnetic symmetry \cite{generalizedglobsym,ethanhigherform,McGreevy22}. At such an interface, we cannot conclude the presence of gapless interface modes from the anomaly, and indeed we show how possible edge modes can get lifted by coupling to the photon field. The lack of edge modes makes it a puzzle to explain what the current carrying states in the superconductor are. We find that again the bulk SPT provides the key, but through a more subtle application of the bulk-boundary correspondence. That is, although the edge is disordered, it still must have an anomaly matching the bulk SPT. The way it does so is that the superconductor-Coulomb ground states are degenerate, with a circle valued parameter we can identify as the phase of a loop order parameter in the Coulomb phase. The anomaly appears as a Thouless pump \cite{Thouless83} as the ground state is adiabatically varied around this circle. In the presence of magnetic fields, trapped by the Meissner effect, this order parameter continually precesses, and drives a persistent current via the Thouless pump.

\subsection{Outline}

We begin with an overview of pure $U(1)$ gauge theory (electrodynamics) in Section \ref{secpureU1}. In Section \ref{subsecconventions} we review some of our conventions, and in Section \ref{subsecmaxwellreview} some features of higher symmetry, including the magnetic symmetry we will use throughout, and features of magnetic symmetry breaking in the Coulomb phase. In Section \ref{subsecpersistenthighercurrent} we discuss some new results concerning persistent higher currents in the Coulomb phase, which follow from the mixed anomaly between the electric and magnetic symmetries, and which we will use later. Those looking to jump right in to the Higgs phase may skip this section and come back to it for reference. 

In section \ref{secU1}, we explore Higgs=SPT phenomena in $U(1)$ gauge theory generally. In Section \ref{subsechiggspstifeldthy}, we study the Landau-Ginzburg theory of the Higgs phase, couple it to background gauge fields for the magnetic and matter symmetries, and then integrate out all dynamical fields to obtain a topological term for the background fields. This topological term classifies the SPT / symmetry-enriched topological phase and controls universal properties of currents and edge/interface modes in these systems. This topological term is the main result of the paper.

In Section \ref{subsechalleffect}, we discuss the boundary anomaly associated with this topological term, and construct candidate edge theories for $d = 2$ and $d = 3$ Higgs phases. We also discuss an interpretation of the Higgs-SPT topological response in terms of a mixed Hall response.

In Section \ref{subseclatticegaugetheory} we give a lattice model for the Higgs-SPT phase. We use the Hamiltonian Villain formalism to have explicit magnetic symmetry (and control monopole fluctuations). This model realizes the edge states of Section \ref{subsechalleffect} on its boundary, and we sketch the phase diagram when the magnetic symmetry is explicitly broken.

In Section \ref{secSC} we use the SPT response and boundary anomaly derived previously to study superconductor phenomenology, in particular Josephson effects and supercurrents.

In Section \ref{subsec:U1SIS} study Josephson effects at superconductor-insulator-superconductor junctions. In the no-tunnelling limit, there is an anomalous interface mode. The anomaly can be seen as an AC Josephson effect. In the presence of weak tunnelling we derive the DC Josephson current from this mode. If there is a matter symmetry distinguishing the condensates on either side, tunnelling is symmetry-forbidden and this interface mode is protected. Our SPT viewpoint naturally points to several measures that can be used to detect this mode and therefore determine the symmetry charge of the condensates.

In Section \ref{subsecU1higgspt} we discuss supercurrents present at a superconductor-insulator interface within the SPT viewpoint. There are no protected edge modes at such interfaces, since the Coulomb/insulator vacuum (spontaneously) breaks the magnetic symmetry needed for protecting such modes. However, the ground state is not unique, and for a superconducting toroid there is a one parameter family of (approximate) ground states parametrized by an angle $\theta$. Although there are no edge modes, we find there is a generalized Thouless pump along the surface in $\theta$, due to the anomaly. When there is magnetic flux through the toroid, $\theta$ precesses activating the Thouless pump which carries the supercurrent. We also discuss analogous phenomena in SPT-SSB interfaces such as an interface between a quantum spin Hall state and a ferromagnet. An interesting outcome of the SPT viewpoint is how it relates seemingly distinct dissipation free current-carrying mechanisms. 

\section{Symmetries and anomalies of pure Maxwell theory}\label{secpureU1}

In this section we will discuss some basic features of pure $U(1)$ gauge theory, its higher symmetries and anomalies, and our conventions. Section \ref{subsecmaxwellreview} is a review of higher symmetries of Maxwell theory, and readers familiar with these topics may safely skip it. The persistent higher currents derived in Section \ref{subsecpersistenthighercurrent} are a new result we will use in deriving the supercurrent in Section \ref{subsecU1higgspt}.

\subsection{Conventions}\label{subsecconventions}

We write the Maxwell action of electromagnetism as: 
\[\label{eqnmaxwellaction}S_{\rm EM} = \frac{1}{2 \mu_0} \int_X d^D x\, (E^2 - B^2) = -\frac{1}{2 \mu_0 e^2} \int_X f \wedge \star f,\]
where in the first expression we have written the usual form of the action in terms of electric and magnetic fields $E$ and $B$, in SI units where $c = \hbar = 1$, $\mu_0$ is the vacuum permeability, and $e$ is the electric charge. We work in signature $-$+++. In the second expression we have introduced the gauge curvature $f = da$, written as the exterior derivative of the gauge field $a$, and the dual field strength $\star da$, where $\star$ is the Hodge star operator \cite{nakahara2018geometry}.

The usual electromagnetic 4-potential (such as in \cite{griffiths2005introduction}) is related to $a$ by $A = \frac{1}{e} a$. This normalization of the gauge field is convenient, so that magnetic flux is quantized in units of $2\pi$, i.e.
\[\oint da \in 2\pi \bZ\]
over closed surfaces. On the other hand, the electric flux is quantized in units of $g^2 := \mu_0 e^2$,
\[\oint \star da \in g^2 \bZ,\]
where we integrate over a closed $(d-1)$-dimensional submanifold (henceforth called a closed $(d-1)$-submanifold).

In $d = 3$ the usual electric and magnetic fields can be recovered as
\[E^i = -\frac{1}{e} g^{ij} (da)_{0j} \qquad B^i = \frac{1}{2e}\epsilon^{ijk} (da)_{jk},\]
where $\epsilon^{ijk}$ is the Levi-Civita symbol and $g^{ij}$ is the spatial metric. In general dimensions, $E^i$ is a vector but $B^{I}$ has $d-2$ anti-symmetric vector indices. For instance in $d = 2$ there is just one component of $B$, which for a two dimensional system embedded in three dimensions is interpreted as the perpendicular component of the 3d magnetic field $B^i$. We will restrict our attention to $d = 2$ and $d = 3$. The canonical commutation relations associated with $S_{\rm EM}$ are (in general dimensions)
\[\label{eqnU1commutationrels}[a_j(x),E^k(y)] = -i \mu_0 e \delta^k_j \delta(x-y).\]
Note that $a_0$ does not have a canonical conjugate: it is a Lagrange multiplier which exists to impose the Gauss law. See \cite{philbin2010canonical,weiberg1996quantum} for more details.

\subsection{Review of higher form symmetry of Maxwell theory}\label{subsecmaxwellreview}

\subsubsection{Higher form symmetries in general}

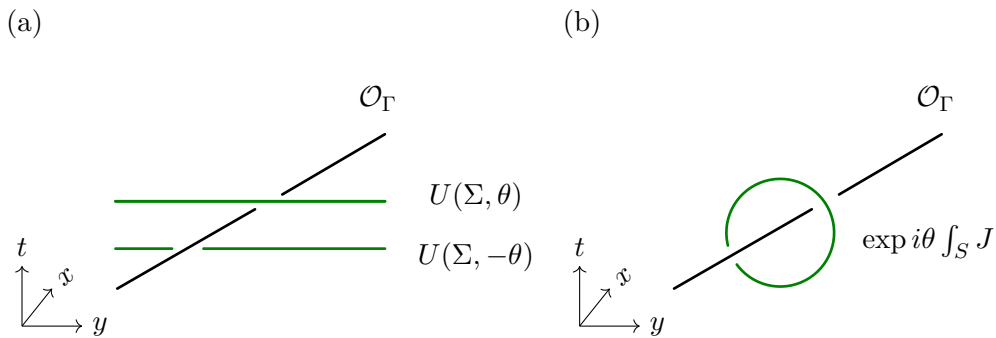
\begin{figure}
    \centering
    \begin{tikzpicture}
    \node at (0,0){\includegraphics[width=5cm]{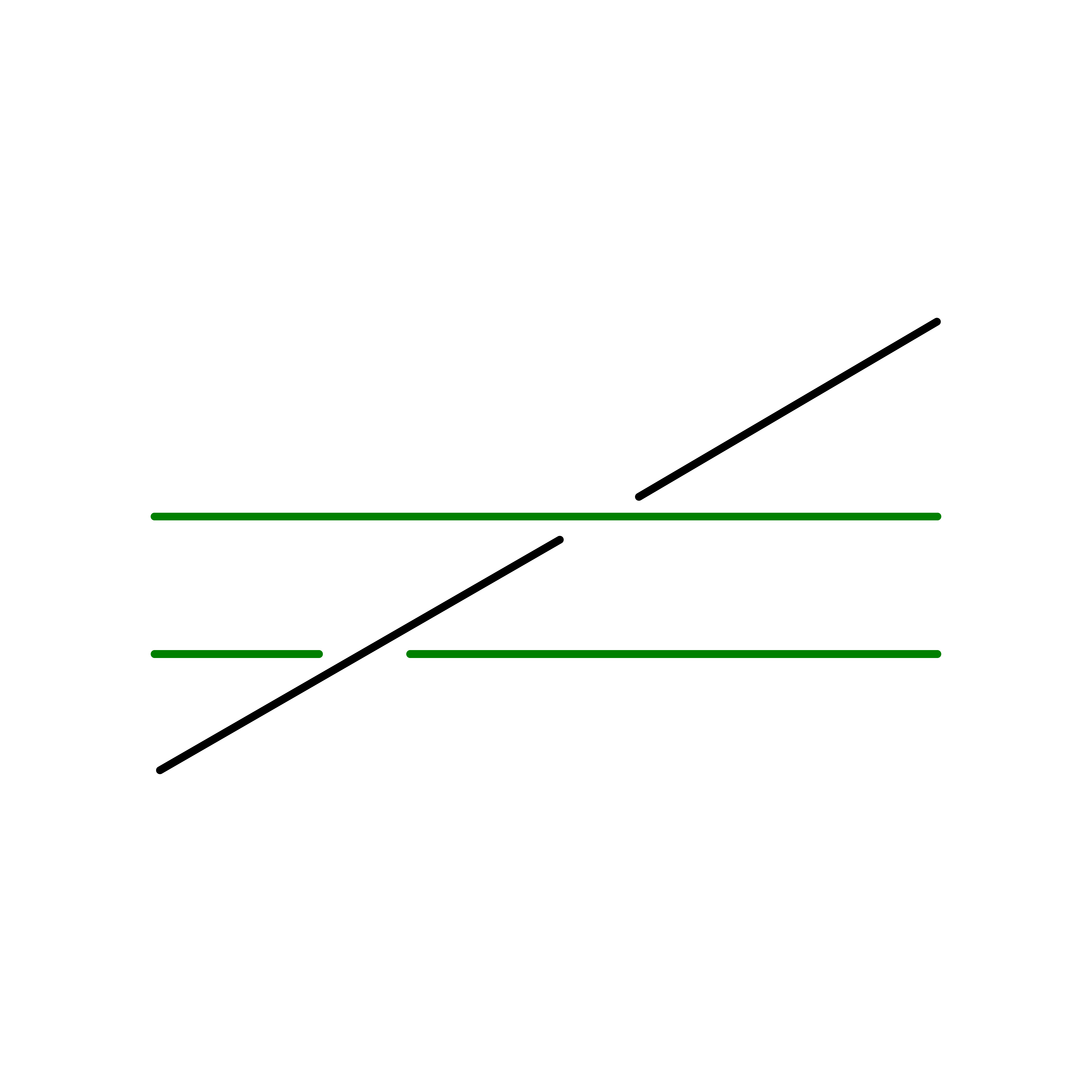}};
    \node at (-3,2.5) {(a)};
    \node at (3,0.2) {$U(\Sigma,\theta)$};
    \node at (1.7,1.5) {$\mathcal{O}_\Gamma$};
    \node at (3,-0.6) {$U(\Sigma,-\theta)$};
        \node at (-2.5,-1){
        \begin{tikzpicture}
        \draw[->] (5,-0.5) -- (5.8,-0.5) node[right] {$y$};
        \draw[->] (5,-0.5) -- (5.4,0.) node[right,rotate=40] {$x$};
        \draw[->] (5,-0.5) -- (5,0.3) node[above] {$t$};
        \end{tikzpicture}
    };
    \end{tikzpicture}
    \begin{tikzpicture}
    \node at (0,0){\includegraphics[width=5cm]{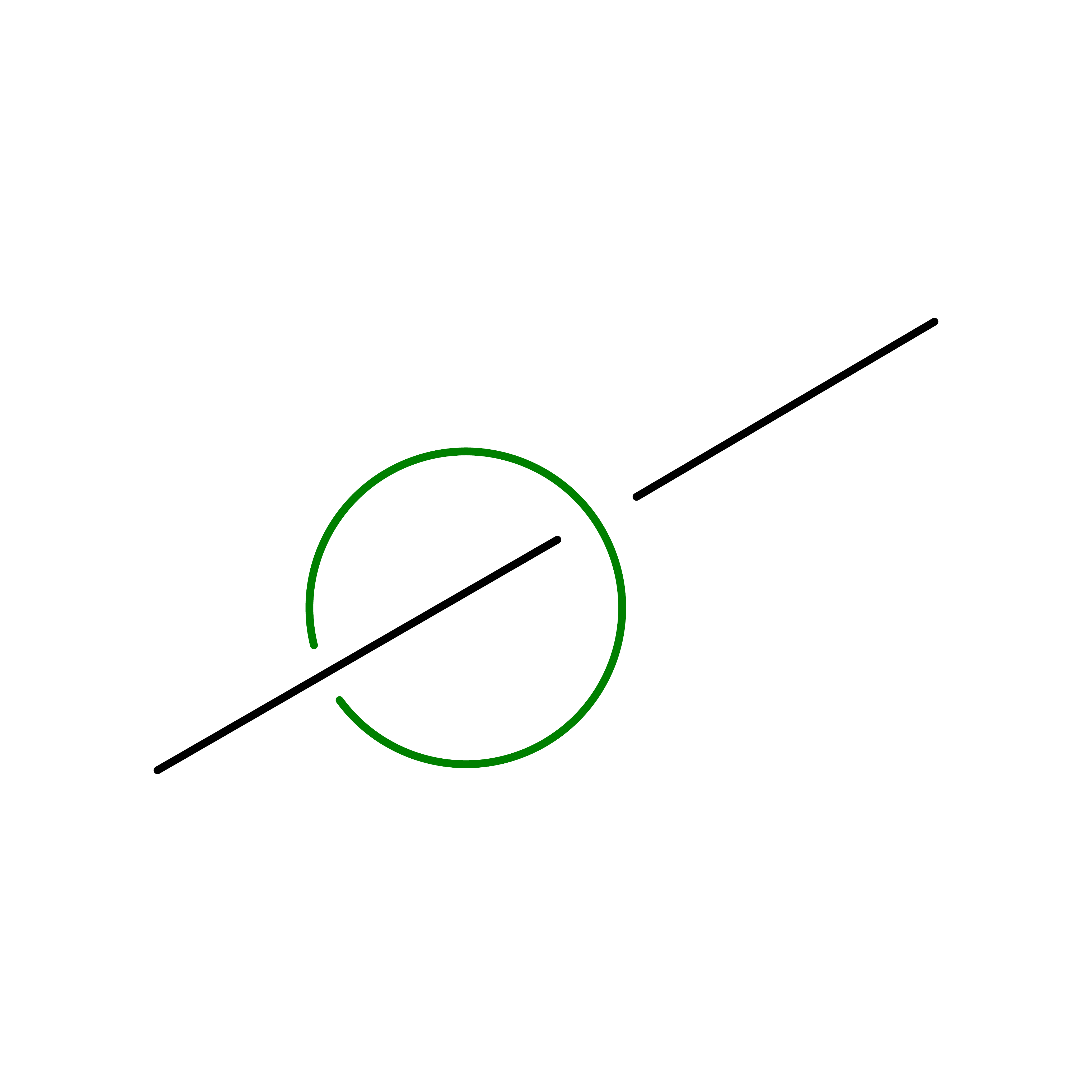}};
    \node at (-3,2.5) {(b)};
    \node at (1.7,1.5) {$\mathcal{O}_\Gamma$};
    \node at (1.6,-0.4) {$\exp i \theta \int_S J$};
        % \node at (3,-0.6) {$U(\Sigma,-\theta)$};
        \node at (-2.5,-1){
        \begin{tikzpicture}
        \draw[->] (5,-0.5) -- (5.8,-0.5) node[right] {$y$};
        \draw[->] (5,-0.5) -- (5.4,0.) node[right,rotate=40] {$x$};
        \draw[->] (5,-0.5) -- (5,0.3) node[above] {$t$};
        \end{tikzpicture}
    };
    % \node at (1.7,2) {$e^{i \phi}$};
    % \node at (-1.5,1) {$Q_{\rm mag}$};
    \end{tikzpicture}
    \caption{\textbf{1-form symmetry action on line operators.} The action of a 1-form symmetry $U(\Sigma,\theta) = \exp \left( i \theta \int_\Sigma J \right)$ (green) on a line operator $\mathcal{O}_\Gamma$ (black) in 2+1d spacetime demonstrated in two ways. In (a), we are looking at a product of operators parallel to the $xy$ plane, evaluated at three different times, with say $\mathcal{O}_\Gamma$ at $t = 0$ and the two symmetry operators at $t=\pm\epsilon$. Since $U(\Sigma,\theta)$ commutes with the time evolution, we can evaluate the whole expression at $t = 0$ to be $U(\Sigma,\theta) \mathcal{O}_\Gamma U(\Sigma,\theta)^{-1}$. Unlike with pointlike operators $\mathcal{O}(x)$, this may yield some non-trivial action of $U(\Sigma,\theta)$ on $\mathcal{O}_\Gamma$, since we cannot deform $\Sigma$ to be disjoint from $\Gamma$. We can think of this arrangement as a linking between $\Sigma$ and $\Gamma$ in spacetime, using the topological invariance of $\int_\Sigma J$ in spacetime to deform (a) into (b).}\label{figwardidentity}
\end{figure}

Let us briefly review some basic facts about higher form symmetries, before turning to Maxwell theory. We refer to \cite{Seiberg_2016,McGreevy22,ethanhigherform} for more details.

In a nutshell, $p$-form symmetries generalize ordinary (i.e., 0-form) global symmetries in that the charges are carried by $p$-dimensional excitations in space (created by $p$-dimensional operators in spacetime). We will mostly be concerned with $p$-form $U(1)$ symmetries, for which there is a conserved Noether current $J$, a $(d-p)$-form; see Appendix \ref{appcurrents} for more details.

\textbf{Symmetry Generators:} For every $(d-p)$-submanifold $\Sigma$ in space we associate a $U(1)$ operator
\[U(\Sigma,\theta) = \exp \left( i \theta \int_\Sigma J \right).\]
We will adopt the conservation law:
\[dJ = 0\]
(in spacetime) which for $p = 0$ is equivalent to the continuity equation $\partial_\mu j^\mu = 0$. For $p > 0$, the mixed time-space components of this equation imply that $\int_\Sigma J$ commutes with the Hamiltonian, while the pure space components imply that if we make a small deformation of $\Sigma$ to some nearby submanifold $\Sigma'$, $\int_\Sigma J = \int_{\Sigma'} J$. The symmetry operators are therefore \emph{topological}: they can be freely deformed as long as we do not encounter any charged operators.

\textbf{Charged Operators:} Since the symmetry operators are topological, for $p > 0$, no local operator can transform under $U(\Sigma,\theta)$. Indeed, say we compute $U(\Sigma,\theta) \mathcal{O}(x) U(\Sigma,-\theta)$. $\Sigma$ may be deformed to some $\Sigma'$ far from $x$, and we have $U(\Sigma,\theta) = U(\Sigma',\theta)$. The latter commutes with any local operator $\mathcal{O}(x)$, so we get $U(\Sigma,\theta) \mathcal{O}(x) U(\Sigma,-\theta) = \mathcal{O}(x)$.

Charged operators must therefore be extended in space, lying along submanifolds $\Gamma$ of dimension $p$ or larger, which intersect $\Sigma$ in an unavoidable way. For this to be possible, either space must have some nontrivial topology, with $\Sigma$ and $\Gamma$ wrapping complementary submanifolds, or we must take $\Sigma$ and $\Gamma$ to be large submanifolds going off to infinity (with careful treatment of the boundary conditions). An example of such a pair is shown in Fig. \ref{figwardidentity}.

\textbf{Symmetry Breaking:} Since there are no local operators which are charged under higher form symmetries, we cannot break them explicitly by adding such operators to the Hamiltonian\footnote{This is not the case in local lattice models like 2+1d toric code, where the 1-form symmetry generators on different curves are really distinct operators in Hilbert space, but only become topological at low energies. In that case there are gapped local operators which are charged. See \cite{partI} for more details.}. Instead, higher form symmetries are broken by introducing new fields which live at the boundary of charged operators $\mathcal{O}_\Gamma$. We will see examples of this below, such as a magnetic monopole at the end of a magnetic flux tube, or a gauge charge at the end of an electric flux tube. These new fields can be added perturbatively by giving them a very large mass, and so the effects of weak symmetry breaking can be studied.

\textbf{Coupling to background fields:} Finally, in the study of all kinds of symmetries it is useful to understand the coupling to background gauge fields. Associated to a $(d-p)$-form current $J$, there is a $(p+1)$-form gauge field $B$ which acts as a current source in the action:
\[\label{eqnbackgroundfields}S = S_0 - \int_{X^{d+1}} B \wedge J.\]
The partition function $Z(X,B)$ with this background is the generating function of current correlation functions, and mostly we will be interested in quantized terms appearing in the derivative expansion of $\log Z(X,B)$. A study of such terms yields a classification of higher form SPTs and anomalies and isolating such terms will allow us to identify the SPT in the Higgs phase.

With this preview, we now turn to a more particular discussion of the higher form symmetries of Maxwell theory.

\subsubsection{Electric 1-form symmetry}\label{subsubsecelecsymm}

First we discuss the electric 1-form symmetry. Although it is not as crucial to our story as the magnetic symmetry, it is a bit simpler to understand, and so we will describe it in some detail, to give an introduction-by-example to higher symmetries.

In the absence of charged matter, the electric field is divergence-free, which follows from the Gauss law
\[d \star da = 0.\]
We interpret this as a conservation law,
\[dJ_{\rm elec} = 0,\]
for the current
\[J_{\rm elec} = \frac{1}{g^2} \star da = \frac{1}{\mu_0 e} \star dA.\]
This current is associated with a conserved charge,
\[\int_\Sigma J_{\rm elec} \in \bZ,\]
for every closed $(d-1)$-submanifold $\Sigma$ in space (surfaces in $d = 3$). By definition, this is the electric flux through $\Sigma$, measured in units of the electric flux quantum $g^2$. We can think of the electric field lines as the charged objects for this symmetry. This flux can only change when an electric charge passes through $\Sigma$. Since we have assumed there is no electrically charged matter, it is constant for all time.

Each conserved charge corresponds to a symmetry generator. In particular, if we study the Hilbert space associated with a spatial slice $Y$, for each $(d-1)$-submanifold $\Sigma \subset Y$, we can define a $U(1)$ operator
\[U(\Sigma,\theta) = \exp \left(i \theta \int_\Sigma J_{\rm elec} \right).\]
Associated to $\Sigma$ is a Poincar\'e dual 1-form $\delta_\Sigma$ (defined up to exact forms), such that
\[U(\Sigma,\theta) = \exp \left(i \theta \int_Y \delta_\Sigma \wedge J_{\rm elec} \right).\]
(See Appendix \ref{apppoincareduality}.) We see from the canonical commutation relations \eqref{eqnU1commutationrels} that this operator acts on the gauge field by
\[a \mapsto a + \theta \delta_\Sigma.\]
Note that although $\delta_\Sigma$ is only defined up to exact forms, the above action is defined up to gauge transformations. In fact, if $\Sigma$ itself is contractible, then $\delta_\Sigma$ is exact, and this is simply a gauge transformation of $a$. More generally, if $\Sigma$ is deformable to $\Sigma'$, then $\delta_\Sigma - \delta_{\Sigma'}$ is exact and so $U(\Sigma,\theta) = U(\Sigma',\theta)$ on the gauge-invariant subspace. We often forget $\Sigma$ and just write the transformation as
\[\label{eqnelec1formsym}a \mapsto a + \lambda,\]
where $\lambda$ is any closed 1-form, hence the name ``1-form symmetry". The two pictures are equivalent, since we can define a (homology class of a) $(d-1)$-submanifold by Poincar\'e dualizing $\lambda$ in $Y$.

All local gauge-invariant operators in pure gauge theory are generated by the field strength $da$, which is invariant under the transformation above, since $d\delta_\Sigma = 0$. However, $U(\Sigma,\theta)$ acts nontrivially on Wilson lines,
\[W(\gamma) = e^{i \int_\gamma a},\]
by
\[W(\gamma) \mapsto W(\gamma) e^{i \theta \int_\gamma \delta_\Sigma} = W(\gamma) e^{i\theta \#(\gamma \cap \Sigma)},\]
where $\#(\gamma \cap \Sigma)$ is the signed intersection number of $\gamma$ and $\Sigma$ in space. Compare Fig. \ref{figwardidentity}(a).

We have said this 1-form symmetry is present in the absence of matter, and in a sense the symmetry is equivalent to this absence. That is, although we cannot break it by adding local operators to the Hamiltonian, we can add a very massive gauge charged field $\phi$. This field allows the Wilson loop to end, since we can consider the gauge invariant combination $\phi(x) e^{i \int_\gamma a} \phi(y)^\dagger$, where $\gamma$ is a path from $x$ to $y$. In the presence of such operators, we can no longer assign consistent symmetry charges to Wilson lines by linking (as in Fig. \ref{figwardidentity}(b)). Moreover, if we recompute $dJ_{\rm elec}$, by the Gauss law we will find this is not zero, but equal to the $\phi$ gauge current.

\subsubsection{Magnetic $d-2$-form symmetry}

There is also a $(d-2)$-form symmetry called the magnetic symmetry, which is associated with the 2-form Noether current
\[J_{\rm mag} = \frac{1}{2\pi} da = \frac{e}{2\pi} dA.\]
We have normalized it so that through a closed surface $\Sigma$,
\[\int_\Sigma J_{\rm mag} \in \bZ\]
counts the integer number of magnetic flux quanta passing through $\Sigma$. The conservation law simply follows from the Bianchi identity
\[dJ_{\rm mag} = \frac{1}{2\pi} d^2 a = 0.\]
This identity holds whenever $a$ is smooth. In this way the magnetic symmetry very analogous with other ``topological symmetries" such as that associated with the winding number of a periodic scalar. More precisely, the magnetic current is conserved as long as there are no monopoles, since by the magnetic Gauss law, $dJ_{\rm mag}$ is the monopole current. The situation is quite similar (in fact dual, see below) to the case of the electric symmetry we discussed above, which is equivalent to the absence of electrically charged matter.

We can associate a symmetry generator with this current for each closed surface $\Sigma$. However, unlike the electric 1-form symmetry, there is no transformation law of $a$ analogous to \eqref{eqnelec1formsym}. Instead, the charged operators for the magnetic symmetry are probe monopoles / 't Hooft operators, which are the analog of vortex operators for a periodic scalar, or more general disorder operators which produce singularities in the fields \cite{Fradkin_2017}.

The probe monopole / 't Hooft operator, $H(\Gamma)$, associated with a closed $(d-2)$-submanifold $\Gamma$ of spacetime, is defined by prescribing $a$ to have a particular singularity along $\Gamma$ \cite{Kapustin_2006}. More precisely, a small neighborhood of $\Gamma$ is removed from spacetime, leaving a boundary which locally looks like $\Gamma \times S^2$, and we impose boundary conditions satisfying
\[\label{eqnthooftline}\int_{S^2} da = 2\pi\]
around the $S^2$ factor. We see by this construction that the monopole operators are magnetically-charged, since if we compute $e^{i \theta \int_{S^2} J_{\rm mag}}$ around a 2-sphere linking the monopole operator, as in Fig. \ref{figwardidentity}(b) (take the green curve to be this 2-sphere), the above boundary condition will yield $e^{i \theta}$.

\subsubsection{Electric-magnetic duality}\label{subsubsecemduality}

In $d = 3$ space dimensions, both the electric and the magnetic symmetries are 1-form symmetries, and there is a duality which exchanges them, along with their charged objects, namely the Wilson and monopole lines \cite{wittenSduality,abeliandualitywalls,readingbetweenthelines,couplingtqftduality}. The duality is between the electromagnetic potential $a$ and a dual $U(1)$ 1-form gauge field $b$. The essential features of this duality follow from the relation
\[\label{eqnemduality}\star da = \frac{g^2}{2\pi} db.\]
The dual field thus has the action
\[S_{\rm EM}' = -\frac{g^2}{8\pi^2} \int_X db \wedge \star db.\]
It has the same quantization conventions as we have been using, with $db$ having $2\pi$ periods and $\star db$ having $\frac{4\pi^2}{g^2}$ periods. The relation \eqref{eqnemduality} can thus be cast as a relation between the 1-form currents:
\begin{align}
J_{\rm mag} = J_{\rm elec}^b, & & J_{\rm elec} = J_{\rm mag}^b,
\end{align}
where $J_{\rm elec}^b$ and $J_{\rm mag}^b$ are defined as in the previous subsections but starting with $b$ instead of $a$. Thus the magnetic symmetry acts as the electric symmetry of the dual gauge field $b$ (i.e., analogously to Eq.~\eqref{eqnelec1formsym}) and vice versa. Furthermore, the $a$ monopole line is the Wilson line for the dual field $b$, since the boundary condition \eqref{eqnthooftline} defining the monopole becomes an electric flux condition
\[\int_{S^2} da = \frac{g^2}{2\pi} \int_{S^2} \star db.\]
So the 't Hooft line has a simple representation $H(\Gamma) = e^{i\int_\Gamma b}$ in the dual variables. 

There is also a $d = 2$ duality, known as particle-vortex duality \cite{peskin1978mandelstam,dasgupta1981phase}. The dual field, which we also call $b$, is a $U(1)$ scalar, and satisfies the same relation \eqref{eqnemduality} but where we now use the $2+1$d Hodge star operator. One sees that the Wilson line of $a$ is the vortex line of $b$ (hence the name ``particle-vortex duality") and the vertex operator $e^{ib(x)}$ is the monopole operator $H(x)$, now a local operator.

\subsubsection{Electric-magnetic symmetry breaking}\label{subsecemssb}

In the usual Coulomb phase (specializing to $d = 3$), the magnetic symmetry is spontaneously broken, in the sense that the monopole operators, which play the role of order parameter, have ``perimeter law" decay \cite{generalizedglobsym,ethanhigherform}
\[\langle H(\Gamma) \rangle \sim \exp(- \alpha |\Gamma|),\]
where $|\Gamma|$ is the length of the loop $\Gamma$. In contrast, in gapped symmetric states like the Higgs phase, this will decay exponentially with the area of a minimal region with boundary $\Gamma$ (``area law"). In the perimeter law case, we can define a renormalized $H(\Gamma)$ by a local counterterm which cancels the exponential decay
\[H'(\Gamma) = e^{ \alpha \int_\Gamma dx} H(\Gamma)\]
so that
\[\langle H'(\Gamma) \rangle \sim C \neq 0.\]
If $\Gamma$ is then taken to be a large loop, i.e. one with a nontrivial action of a magnetic 1-form symmetry generator, this symmetry rotates the complex phase of $C$, and therefore the argument of $C$ labels a circle of degenerate vacua breaking this symmetry.

{\bf Monopole line perimeter law:} Let us demonstrate the perimeter law of the monopole line. To study $\langle H(\Gamma) \rangle$ we remove an $\epsilon$-neighborhood of $\Gamma$ from spacetime $X$, and study the path integral on the resulting space $X_\epsilon$. The field $a$ may be decomposed as $a = a_0 + a_1$ where $a_0$ is fixed and has monopole boundary conditions along $\partial X_\epsilon$, while $a_1$ has a smooth extension to $X$ and is integrated over. We compute
\[\langle H(\Gamma) \rangle = \frac{\int D a_1 e^{ -S(a_0 + a_1)}}{\int D a_1 e^{-S(a_1)}} =e^{\frac{1}{2g^2} \int_{X_\epsilon} f_0 \wedge \star f_0} \langle e^{\frac{1}{g^2} \int_{X_\epsilon}  f_0 \wedge \star f_1} \rangle.\]
If we choose $a_0$ to be a solution to the Maxwell equations of motion, then the expectation value on the right hand side is 1 (this is a special feature of the quadratic form of the Maxwell action). Thus the value of the 't Hooft loop is equal to the minus exponentiated action of a classical Dirac monopole solution. This will diverge in the $\epsilon \to 0$ limit, but we are really only interested in the long range behavior of the 't Hooft loop.

For concreteness, let us take $X = \bR^3 \times S^1$, with the 't Hooft operator inserted at the origin in $\bR^3$, i.e. $\Gamma = \{0\} \times S^1 \subset X$, and the $S^1$ of length $L$. The Dirac solution has $B \sim 1/r^2$ in $\bR^3$, constant along $S^1$. This corresponds to an action $S \sim L \int d^3 x \frac{1}{r^4}$. The $\bR^3$ integral is infrared finite, so we find $\langle H(\Gamma)\rangle \sim e^{- \alpha L}$ for some constant $\alpha$. This is the perimeter law, so we conclude the magnetic symmetry is spontaneously broken. A similar argument can be made in all dimensions $d \ge 2$ (see \cite{ethanhigherform} Section 5.1.2).

{\bf Tower of States:} On a finite space, such as a 3-torus, the symmetry breaking degeneracy of electromagnetism is lifted to $1/L$. Generally speaking for broken $p$-form symmetries, there are two types of low lying states, which one can think of as momentum and winding modes of the order parameter, which in low energies we can characterize by a $p$-form gauge field $b$, related to the order parameter by
\[H(\Gamma)/|H(\Gamma)| = e^{i \int_\Gamma b}.\]
For Maxwell theory, $b$ is the dual field to $a$ and the two towers of low lying states correspond to static electric and magnetic field states, respectively.

In general, the Hamiltonian for $b$ to leading order will be the generalized Maxwell Hamiltonian
\[H = \int d^d x \left(E^2 + B^2\right),\]
where $E_J = (db)_{0J}$ is the ``electric field'' of $b$ and $B_K = (db)_K$ is the ``magnetic field" of $b$, where $J$ and $K$ are antisymmetric spatial multi-indices degree $p$ and $p+1$ respectively. Electric flux is quantized through $(d-p)$-submanifolds $Y^{d-p}$, which gives the ``momentum states''
\[\int_{Y^{d-p}} \star db \in k \bZ,\]
where $k$ is some electric flux quantum and $\star$ is the spacetime Hodge star (for $p = 0$ these quantized momentum states will give the usual tower of states). And magnetic flux is quantized around $(p+1)$-submanifolds $Z^{p+1}$, which gives the ``winding states''
\[\int_{Z^{p+1}} db \in 2\pi \bZ.\]
These are states where the order parameter is constant in time but winds around the cycle $Z^{p+1}$.

Studying the energy of these states, the momentum states contribute the usual Anderson tower $E \sim 1/L^{d-2p}$ and the winding states contribute a dual tower $E \sim 1/L^{2p+2-d}$. The momentum Anderson tower is present so long as $d \ge 2p+1$, while the winding Anderson tower is present so long as $d \le 2p+ 1$ (which is absent in usual symmetry breaking $p = 0$, $d > 1$). Only in the special dimensions where a $p$-form field is self dual, $d = 2p + 1$, are both present, with $1/L$ splitting.

{\bf Generalized Goldstone modes:} One can also identify Goldstone modes associated with this broken symmetry, which turn out to be the soft photon states \cite{kovner1991photon}. To illustrate this, let $\ket{p,\xi}$ be a state created from a reference vacuum $\ket{0}$ by adding a single photon of momentum $p^\mu$ and polarization $\xi^\mu$, ie. $\ket{p,\xi} = \xi^\mu a_\mu(p)\ket{0}$. From $f_{\mu \nu} = \partial_\mu a_\nu - \partial_\nu a_\mu$ we can derive the following matrix element:
\[\bra{p,\xi} f_{\mu \nu}(x) \ket{0} = (p_\mu \xi_\nu - p_\nu \xi_\mu) e^{-i p_\mu x^\mu}.\]
This is analogous to how usual Goldstone modes are created by the charge density in broken 0-form symmetries. This matrix element appears in the calculation of the order parameter via spectral decomposition \cite{weiberg1996quantum}, and it can be derived directly from the symmetry breaking order parameter, see \cite{ethanhigherform}. It appears that magnetic symmetry breaking may in fact be equivalent to the presence of photon states, but we have not been able to prove this (although see \cite{Hofman_2019}).

\subsection{Electric-Magnetic anomaly and persistent higher currents}\label{subsecpersistenthighercurrent}

Following the review of the previous sub-section, we will now discuss the electric-magnetic anomaly that is present in the Coulomb phase and point out a consequence of it, which will be important in our derivation of persistent currents in a superconductor. While the result we derive can be obtained in different ways, we cast it in the language of ``persistent currents'', since the result is a higher form symmetry version of persistent currents of 0-form symmetries. First however, we will introduce the background fields that couple to the higher form symmetries and relate them to some familiar quantities.

\subsubsection{Sources, polarization, magnetization}\label{subsec:sources}

We can couple Maxwell theory to background gauge fields for the higher form symmetries. In general a $p$-form $U(1)$ symmetry with $(d-p)$-form current $J$ can be coupled to a $(p+1)$-form background gauge field\footnote{See Appendix \ref{apphigherformproperties} for a mathematical definition of these objects.} $B$ via a term $B \wedge J$, as in Eq.~\eqref{eqnbackgroundfields}, so that $B$ acts as a current source, that can be used to get the generating function for $J$ correlation functions.

Especially relevant for us are the timelike (or mixed) components of $B$, which define $p$-form chemical potentials
\[\mu_I := B_{0I}.\]
Let us write the higher form charge density
\[\rho^I = \frac{1}{(d-p)!}\epsilon^{IJ} J_J.\]
The chemical potentials modify the Hamiltonian by
\[H = H_0 - \int_Y \mu \wedge J = H_0 - \int_Y \mu_I \rho^I d^dx.\]
If $\mu$ is closed (which is required for these terms to commute), we can find a Poincar\'e dual $(d-p)$-cycle $\Sigma$ (i.e., a formal integer combination of $(d-p)$-submanifolds) and a real number $u$ satisfying
\[\int_Y \mu \wedge J = u \int_\Sigma J.\]
We recognize $\int_\Sigma J$ as the conserved $J$-charge measured along $\Sigma$. We see that $\mu$ shifts the energy of eigenstates of $H_0$ an amount $u$ proportional to this charge, just like an ordinary chemical potential.

In electromagnetism in $d = 3$, the chemical potentials $\mu^{\rm elec}$ and $\mu^{\rm mag}$ associated to $J_{\rm elec}$ and $J_{\rm mag}$ have a familiar form. Let us first note that the charge densities are
\[\rho_{\rm elec}^i = \frac{1}{(d-1)!}\epsilon^{iJ} J_{{\rm elec},J} = \frac{1}{\mu_0 e} E^i \qquad \rho_{\rm mag}^i = \frac{1}{(d-1)!}\epsilon^{iJ} J_{{\rm mag},J} = \frac{e}{2\pi} B^i.\]
Thus the chemical potentials are the sources of $E$ and $B$ fields:
\[H = H_0 - \int d^3x \left(\frac{1}{\mu_0 e} \mu_i^{\rm elec} E^i + \frac{e}{2\pi} \mu_i^{\rm mag} B^i \right).\]
We can consider the chemical potentials as resulting from some background magnetization $\vec M$ and polarization $\vec P$ vectors:
\[\frac{1}{\mu_0 e} \mu^{\rm elec}_j = P_j \qquad \frac{e}{2\pi} \mu^{\rm mag}_j = M_j.\]
This identification can be checked by deriving Maxwell's equations in their presence (for the derivation of which we can forget any flatness conditions on the $\mu$'s)
\[df = 0 \\
\frac{1}{g^2} d \star f = d \mu^{\rm mag} + d \star \mu^{\rm elec}.\]
We note that while the magnetic higher current $J_{\rm mag}$ is still conserved\footnote{This is a result of our choice to use a smooth gauge field $a$. Implicitly this means that $\vec M$ and $\vec P$ are the result of electric rather than magnetic matter.}, the electric higher current is not:
\[\label{eqnelectricnonconserv1}dJ_{\rm elec} = d\mu^{\rm mag} + d\star \mu^{\rm elec}.\]
This is a feature of the electric-magnetic anomaly we now discuss.

\subsubsection{Electric-magnetic anomaly and persistent higher currents}\label{subsubsecemanomaly}

The electric and magnetic symmetries share an anomaly rather analogous to the chiral anomaly of the so called ``momentum'' and ``winding'' symmetries of a 1+1d compact boson. We have already seen a hint of this in the violation of the conservation laws in the presence of background fields in Eq.~\eqref{eqnelectricnonconserv1}. Another way to look at the anomaly is as a contact term in the commutator of the charge densities, which we compute using the canonical commutation relations \eqref{eqnU1commutationrels}:
\[\label{eqnEManomcomm}[\rho_{\rm mag}^I(x), \rho_{\rm elec}^j(y)] = \frac{i}{2\pi} \epsilon^{kIj} \partial_k \delta(x-y).\]
These two features are related to each other, as we review in Appendix \ref{appanomrelationships}.

The main use of the anomaly for us is that it gives rise to \emph{equilibrium (higher) currents}, in the sense that the expected currents $\langle J_{\rm mag} \rangle$ and $\langle J_{\rm elec}\rangle$ need not vanish. As noted in \cite{Elsecritdrag} in a different context, although Bloch's theorem would seem to rule out such currents, a local anomaly like the above provides a loophole in the presence of background gauge fields. Here, the generalized currents of the higher form electric and magnetic symmetries will simply be electric and magnetic fields which cannot relax to zero, due to the absence of monopoles and charges at low energies, which would otherwise screen the fields by pair creation \cite{schwinger1951gauge}. Note that in $d=2$, particle-vortex duality between the Coulomb phase and a superfluid maps the magnetic current to the persistent particle current.

We give a general derivation of persistent currents in Appendix \ref{apppersistcurrents}, for all dimensions and symmetry types. For now let us simply give the argument for electromagnetism in $d = 3$. The general method follows \cite{Elsecritdrag}. Consider the Gibbs free energy in the presence of the 1-form chemical potentials:
\[H = H_0 - \int d^3 x \left( \mu_k^{\rm elec} \rho^k_{\rm elec} + \mu_k^{\rm mag} \rho^k_{\rm mag} \right).\]
We take $d\mu^{\rm elec} = d\mu^{\rm mag} = 0$, so that all terms commute, and we are in equilibrium\footnote{In terms of the background magnetization and polarization, they are time independent and curl free, $\nabla \times M = \nabla \times P = 0$.}. We will compute the magnetic current in a pure electric background ($\mu^{\rm mag} = 0$). The other case follows the same argument. Because we are in equilibrium, the ground state\footnote{Here we focus on $T=0$; extension to finite temperature is straightforward~\cite{Elsecritdrag}.} $|0\rangle$ minimizes the expectation value of $H$. We can consider varying the state according to
\[|\epsilon\rangle = \left(1 + i \epsilon \int d^3 x \eta(x)_j \rho_{\rm mag}^j(x)\right)|0\rangle\]
for some test form $\eta_j$ and some small $\epsilon$. Because $\langle 0|H|0\rangle$ is a minimum, $\langle \epsilon | H |\epsilon \rangle = \langle 0|H|0\rangle$ to order $\epsilon$. This means
\[0 = \int d^3 x \, \eta_j(x) \langle 0| [H_0,\rho^j_{\rm mag}(x)]|0\rangle - \int d^3 x \, d^3 y \, \mu_k^{\rm elec}(x) \eta_j(y) \langle 0| [\rho_{\rm elec}^k(x),\rho^j_{\rm mag}(y)]|0 \rangle
\\ \equiv (1) - (2). \notag\]
We can rewrite the first term using
\[[H_0,\rho^j_{\rm mag}] = -i\partial_0 \rho^j_{\rm mag} = -\frac{1}{2} i\epsilon^{jkl} \partial_0 J^{\rm mag}_{kl} = \frac{1}{2} i \epsilon^{jkl}(\partial_{k} J^{\rm mag}_{0l} - \partial_{l} J^{\rm mag}_{0k}) = i \epsilon^{jkl} \partial_k J^{\rm mag}_{0l},\]
where we used the conservation law $dJ^{\rm mag} = 0$, which holds for the Hamiltonian $H_0$ that has no chemical potentials in it. We can now integrate the first term by parts to write
\[(1) = - i\int d^3 x \epsilon^{jkl} (\partial_k \eta_j) \langle J^{\rm mag}_{0l} \rangle.\]
Meanwhile for the second term, we use the anomalous commutator \eqref{eqnEManomcomm}
\[\langle 0| [\rho_{\rm elec}^k(x),\rho^j_{\rm mag}(y)]|0 \rangle = \frac{i}{2\pi} \epsilon^{lkj} \partial_{l} \delta(y-x),\]
so
\[(2) = \frac{i}{2\pi} \int d^3 x d^3 y \mu_k^{\rm elec}(x) \eta_j(y) \epsilon^{lkj} \partial_l \delta(y-x).\]
Doing the $y$ integral we obtain
\[(2) = \frac{i}{2\pi}\int d^3 x \epsilon^{lkj} (\partial_l \eta_j) \mu_k^{\rm elec} = \frac{i}{2\pi} \int d^3 x \epsilon^{jkl} (\partial_k \eta_j) \mu_l^{\rm elec}.\]
We have $(1)=(2)$. Since $\eta$ was arbitrary, we can equate the integrands, and find
\begin{align}
\label{eqnempersistentcurrents}-\langle J^{\rm mag}_{0j} \rangle = \frac{e}{2\pi} \langle E_j \rangle = \frac{1}{2\pi}\mu_j^{\rm elec}, & & -\langle J^{\rm elec}_{0j} \rangle = \frac{1}{\mu_0 e}  \langle B_j \rangle = \frac{1}{2\pi}\mu_j^{\rm mag},
\end{align}
where the second equation comes from considering the opposite case, with $\mu^{\rm elec} = 0$, $\mu^{\rm mag} \neq 0$. Note that these hold only up to total derivatives. (In fact, the currents themselves are only defined up to total derivatives, see Appendix \ref{appcurrents}.) Since they are derived from the anomaly, these equations also hold for higher derivative deformations of $S_{\rm EM}$, although the form of $J^{\rm elec}$ will change.

\section{$U(1)$ Higgs Phases\label{secU1}}

In this section we present our main results on $U(1)$ Higgs-SPT phenomena, and in Section~\ref{secSC} we will describe their relationship with superconductors. In particular, Section \ref{subsechiggspstifeldthy} derives the SPT response from Landau-Ginzburg (LG) theory (from that point on we will not need to invoke LG theory again). We show in Section \ref{subsechalleffect} that the quantized SPT response implies a boundary anomaly shared between a $0$-form ``matter symmetry" and the $d-2$-form magnetic symmetry. This specific boundary anomaly of a superconductor has not, to the best of our knowledge, been previously pointed out. Furthermore, given that it is a renormalization group invariant, it places non-perturbative constraints on the boundary dynamics at low energies, which we discuss. In Section \ref{subseclatticegaugetheory} we discuss a lattice model where this boundary anomaly and the associated edge modes can be seen concretely.

\subsection{Landau-Ginzburg Derivation of SPT Response}\label{subsechiggspstifeldthy}\label{subsec:Higgsyboi}\label{subsubsectionhall}

A useful model for thinking about gapped Higgs phases of $U(1)$ gauge theory is the Landau-Ginzburg theory/abelian Higgs model of a charge $m$ complex scalar\footnote{The Higgs field $\Phi$ itself could be a bound state of more fundamental fields, such as in a superconductor, where $m = 2$ and $\Phi$ is the Cooper pair field.} $\Phi$ and a gauge field $a$,\footnote{In the case with a spin-charge relation, we should treat $a$ as a $Spin^c$ connection. The ways to do this are relatively standard (see for instance \cite{seibergwittenTI}) and will not affect our conclusions significantly, so we will mostly ignore this subtlety. However, see the discussion around Eq.~\eqref{eqnCSaction} below.} with the action
\[\label{eqnabelianhiggs}S_{\rm LG} = S_{\rm EM} + \int_X d^{d+1}x \left(|(\partial_\mu - i m a_\mu )\Phi|^2 + t |\Phi|^2 + u |\Phi|^4 \right),\]
where $u > 0$, $t$ is a tuning parameter that can be positive or negative, and $S_{\rm EM}$ is the Maxwell action \eqref{eqnmaxwellaction}. For simplicity we have assumed a Lorentz invariant form of the action, although this is by no means necessary for what follows. The Higgs phase occurs when $t \ll 0$, leading to the condensation of $\Phi$. In this phase, we can write $\Phi = r e^{i \phi}$, and $r$ will have gapped fluctuations around its nonzero minimum. We will derive the SPT response in this effective field theory, although the SPT response is quantized, so it holds in the entire Higgs phase.

We are interested in topological responses, for which we can take all mass parameters (i.e., both the Higgs and the photon mass) to infinity. We can do this by taking $t \to - \infty$ while keeping $u$ and $e$ finite. Doing so, we obtain a fixed point limit of the Higgs phase which is a simple constraint\footnote{Note the analogy with the SPT stabilizer $Z \sigma^z Z = 1$ discussed for the $\mathbb{Z}_2$ gauge theory in Part I (Sec. 4.2 \cite{partI}): both imply that matter defects (domain walls for $\mathbb{Z}_2$ and vortices for $U(1)$) are charged under the appropriate magnetic symmetry. The Gauss law also plays a similar role in both cases, ensuring that defects of the gauge fields are charged under the matter symmetry.} $d\phi = ma$. We can write it with a Lagrange multiplier field $\lambda$ \cite{gukov2013topological}:
\[\label{eqnhiggsaction}S_{\rm IR} = \int_X  \lambda \wedge (d\phi - ma) .\]
The gauge symmetry acts as
\[
a \mapsto a + dg \\
\phi \mapsto \phi + mg \\
\lambda \mapsto \lambda.
\]

To derive the SPT response, we will couple this theory to background gauge fields and compute the effective action. Recall we have the conserved magnetic current $J_{\rm mag} = da$. The presence of electrically charged fields in $S_{\rm LG}$ does not spoil this conservation law. We can include a minimal coupling to a background $(d-1)$-form gauge field $B_{\rm mag}$ by simply adding the term
\[\label{eqnmagcscoupl}S_{\rm mag} = -\int_X B_{\rm mag} \wedge \frac{da}{2\pi}\]
to the action $S_{\rm LG}$ or $S_{\rm IR}$.

Suppose also that $\phi$ has charge $q$ under some additional \emph{global} matter symmetry $U(1)_{\rm mat}$\footnote{We can also take the matter symmetry to be a cyclic subgroup of this $U(1)$.} and we likewise couple to a background 1-form gauge field $A_{\rm mat}$. This is physically reasonable so long as there are gauge-invariant $U(1)_{\rm mat}$ charged states in the theory. This can arise for instance when there are other very massive fields with gauge and global charges not commensurate with $\phi$ (so some combination with $\phi$ is gauge invariant), when the gauge symmetry itself is emergent, or in the presence of a no-tunneling defect (see Section \ref{subsubsecweaktunnelingjunction}). All these situations were explored and discussed in detail for the case of $\mathbb Z_2$ gauge theory in Part I \cite{partI}. To be totally precise, we could include some massive fields in $S_{\rm IR}$ which are charged under $a$ but not $A_{\rm mat}$ (or vice versa). These are spectator fields, and just exist to distinguish $U(1)_{\rm mat}$ from the gauge symmetry. They will not condense in the Higgs phase, and they do not change the derivation. (However, they will have to come down in energy at a boundary, see below.)

In the presence of the $U(1)_{\rm mat}$ background field $A_{\rm mat}$, the covariant derivative $d\phi - ma$ in \eqref{eqnhiggsaction} must be modified to $d\phi - ma - qA_{\rm mat}$, so the full action of the fixed point with both background fields is
\[S_{\rm IR}' = \int_X \left( \lambda \wedge (d\phi - ma - qA_{\rm mat}) -  B_{\rm mag} \wedge \frac{da}{2\pi} \right).\]
Now, integrating out $\lambda$ yields the constraint
\[\label{eqnhiggsconstraintwbg}d\phi - ma - qA_{\rm mat} = 0.\]
When $m = 1$, this can be inverted to give $a = d\phi - q A_{\rm mat}$. Plugging this back into the action, the coupling to magnetic background gives
\[\label{eqnU1higgsbackgroundterm}S_{\rm SPT} = q \int_X B_{\rm mag} \wedge \frac{dA_{\rm mat}}{2\pi}.\]
A topological term such as this in the effective action for the background gauge fields characterizes the SPT response, and we find this is a nontrivial SPT iff $q \neq 0$, that is as long as the condensate carries global charge under $U(1)_{\rm mat}$. In $d=2$ this phase is equivalent to a bosonic quantum spin-Hall phase \cite{Kane05,Bi17}, protected by $U(1)_{\rm mat} \times U(1)_{\rm mag}$. In $d \geq 3$, it is a non-trivial generalization, where one of the protecting symmetries becomes higher-form. This topological response is one of our main results and will be used repeatedly below.

When $m > 1$, we cannot invert the constraint equation \eqref{eqnhiggsconstraintwbg} for $a$ because it is insensitive to shifts of $a$ by a flat $\bZ_m$ gauge field. In flat space we can still formally integrate out $\phi$ to find the fractional SPT response
\[\label{eqnU1higgsbackgroundtermfract}S_{\rm ``SPT"} = \frac{q}{m} \int_X B_{\rm mag} \wedge \frac{dA_{\rm mat}}{2\pi}.\]
To understand this more precisely, we should consider these partial Higgs phases as a symmetry-enriched topological (SET) phase rather than an SPT. However, all the essential physics we will need is contained in \eqref{eqnU1higgsbackgroundtermfract}. For completeness, we can describe the SET by integrating out $\phi$, which yields the constraint $\lambda = dw$ for some $U(1)$ gauge field $w$ (Wilson lines for $w$ are vortex lines for $\phi$). This leads to the Chern-Simons action
\[\label{eqnCSaction}
S_{\rm CS} = \int_X \left( \frac{m}{2\pi} dw \wedge a + q \frac{dw}{2\pi} \wedge A_{\rm mat} - B_{\rm mag} \wedge \frac{da}{2\pi} \right).
\]
This describes a $\bZ_m$ topological order, enriched by matter and magnetic symmetries. The topological order was pointed out previously in Ref.~\cite{hansson2004superconductors}; here, we mainly focus on the aspect having to do with global symmetries, which are already present in the $m=1$ case.

\subsubsection*{Chern-Simons theory and fundamental fermions}

In a $2e$ superconductor, the condensate is formed by paired charges ($m=2$). However, there is an important subtlety in the interpretation of this action when the fundamental charge $e$ excitations are fermions, such as in a physical superconductor. In particular, treating \eqref{eqnCSaction} naively we will find that the quasiparticle associated with the charge $e$ particle, namely the unit $a$ Wilson line, is necessarily a \emph{boson}. Meanwhile, the emergent fermion in \eqref{eqnCSaction} carries magnetic flux, and cannot be identified with the electron.

The correct resolution is to treat $a$ as a Spin$^c$ connection, which we must do anyway to have $a$-charged fermions. A Spin$^c$ connection is very similar to a $U(1)$ gauge field, except it has the modified magnetic flux quantization condition:
\[\int_\Sigma \frac{da}{2\pi} + \frac{1}{2} w_2(TX) \in \bZ,\]
where $w_2(TX)$ is the 2nd Stiefel-Whitney class of the tangent bundle $TX$ of spacetime $X$, and $\Sigma$ is an arbitrary closed surface in $X$.

Treating $a$ as a Spin$^c$ connection essentially modifies the spin of every $a$-charged field by $1/2$ \cite{seibergwittenTI}, so the Wilson line of $a$ becomes a fermion and the bound state of the $a$ and $w$ charges (the latter being the $\pi$ flux quasiparticle) is a boson. We can see this as follows. On a manifold with a spin structure $\eta$, which is a $\bZ_2$ gauge field with the modified flux condition
\[d\eta = w_2(TX) \mod 2,\]
we can relate $a$ to an ordinary $U(1)$ gauge field $a'$ by
\[a' = a + \pi \eta.\]
Writing \eqref{eqnCSaction} in terms of $a'$ and $\eta$ (and turning off background fields) we obtain
\[S_{\rm spin\ CS} = \int_X \frac{m}{2\pi}  dw \wedge a' + \frac{m}{2} dw \wedge \eta.\]
Note $m$ must be even for this to make sense (only then can the charge $m$ condensate be bosonic). The coupling $\frac{m}{2} dw \wedge \eta$ has the effect of modifying the $2\pi/m$ $w$ flux, i.e. the unit $a'$ charge, to be a fermion \cite{Thorngren_2015}. 

Note that the (fractional) SPT response is not modified, because the quantum numbers of the Higgs condensate are assumed to be the same, whether it is made from bosons or fermions.  However, the coupling to the magnetic background field $- \frac{1}{2\pi} B_{\rm mag} \wedge da$ is half-quantized if $a$ is a Spin$^c$ connection and $dB_{\rm mag}$ has $2\pi \bZ$ integrals. This is associated with an anomaly of the form\footnote{We are grateful to Max Metlitski for pointing out this subtlety.}
\[S_{\rm spin\ anom} = \frac{1}{2} \int_{Z^5} w_2(TZ) \wedge \frac{dB_{\rm mag}}{2\pi}.\]
The physical meaning of this anomaly is that if we gauge the magnetic symmetry, we effectively turn off $a$, and the fluxes for the magnetic symmetry become local, neutral fermions signalling an anomaly. We can thus cure the anomaly by introducing gauge neutral fermions to begin with, so we are not working within an intrinsically bosonic Hilbert space. This was the approach taken in $d=2$ in \cite{Moroz_2017}.

\subsection{Boundary anomaly and mixed Hall effect}\label{subsechalleffect}

The topological responses \eqref{eqnU1higgsbackgroundterm}, \eqref{eqnU1higgsbackgroundtermfract} are not gauge invariant on manifolds $X$ with boundary. For example, performing a magnetic gauge transformation
\[\label{eqnanomalousvariationmag}B_{\rm mag} \mapsto B_{\rm mag} + d\lambda \\ 
\delta S_{\rm ``SPT"} = \frac{q}{m} \int_{\partial X} \lambda \wedge \frac{dA_{\rm mat}}{2\pi}.\]
Meanwhile, performing a matter gauge transformation, we find
\[\label{eqnanomalousvariationmat}A_{\rm mat} \mapsto A_{\rm mat} + d\eta \\ 
\delta S_{\rm ``SPT"} = \frac{q}{m} \int_{\partial X} \eta \wedge \frac{dB_{\rm mag}}{2\pi}.\]
To cancel this variation (even for the SPT case $m = 1$) there must be additional modes on the boundary with an anomaly, transforming in such a way to cancel this boundary variation. We will construct possible boundary field theories in $d = 2$ and $d = 3$ and show that they cancel the above variations. A lattice model giving rise to each of these theories will be discussed in Section \ref{subseclatticegaugetheory}.

For Higgs-SPT, we must be very careful about the interpretation of these edge modes, as this abstract construction applies only to an interface to a symmetry-preserving trivial phase. In an exact gauge theory (where the Gauss law is never violated), this symmetry-preserving trivial phase must also be a Higgs phase, but where the Higgs field carries trivial matter charge $q = 0$. This gives an interpretation of the interface mode as the relative phase of the Higgs condensates $\phi_L - \phi_R$, which is gauge invariant. In Section \ref{subsec:U1SIS} we will use this point of view to discuss SPT interface physics in Josephson junctions.

In emergent gauge theories (where the Gauss law is energetically favored but not a UV constraint), it is possible to consider an interface to a symmetric product state. In that case, the Higgs-SPT phase is essentially just an SPT phase, and we expect our derivation to apply as usual.

On the other hand, when we consider interfaces between a superconductor and the ``vacuum'', which is in fact the Coulomb phase, one must exercise caution. The Coulomb phase spontaneously breaks the magnetic symmetry, so the usual argument for edge modes does not apply, and in Section \ref{subsecU1higgspt} we will show the edge modes described in this section are in fact unstable to coupling to the electromagnetic field outside the superconductor.

\subsubsection{$d = 2$ boundary Luttinger liquid}

In $d = 2$, the anomalies~\eqref{eqnanomalousvariationmag}, \eqref{eqnanomalousvariationmat} can be cancelled by a boundary Luttinger liquid described by two canonically conjugate scalar fields $\theta$ ($2\pi$ periodic) and $\varphi$ ($2\pi m$ periodic with residual $\bZ_m$ gauge symmetry $\varphi \mapsto \varphi + 2\pi$), with the Lagrangian
\[\label{2dtheory}\mathcal{L}_{\rm 1d\ edge} =  \frac{1}{2\pi m} d\varphi \wedge d\theta + \cdots,\]
where $\cdots$ denotes nonquantized couplings.
The symmetries act by
\begin{align}
U(1)_{\rm mat} : \begin{cases} \varphi \mapsto \varphi + q \alpha \\ \theta \mapsto \theta \end{cases}, & & U(1)_{\rm mag} : \begin{cases} \varphi \mapsto \varphi \\ \theta \mapsto \theta + \beta \end{cases}.
\end{align}
These symmetries are well-known to have a mixed anomaly (the ``chiral'' or ``axial'' anomaly) matching $S_{\rm ``SPT"}$.

We can see the anomaly as follows. The kinetic term corresponds to a canonical commutation relation
\[\label{eqncmpboscom}[\varphi(x),\partial_y \theta(y)] = 2\pi i m \delta(x-y).\]
From the above symmetry action, the $U(1)_{\rm mat}$ current is thus
\[\label{eqnnaivematcurrent2d}J_{\rm mat,\partial} = \frac{q}{2\pi m} d\theta,\]
so minimal coupling to matter background is
\[S_{\rm mat} = - \frac{q}{m} \int A_{\rm mat} \frac{d\theta}{2\pi}.\]
Under a magnetic symmetry gauge transformation, $\theta \mapsto \theta + \lambda$, so we see the variation of $S_{\rm mat}$ cancels $\delta S_{\rm ``SPT"}$ in \eqref{eqnanomalousvariationmag}. We can likewise show the same for \eqref{eqnanomalousvariationmat}.

\subsubsection{$d = 3$ boundary superfluid}\label{subsubsectionboundarysuperfluid}

In $d = 3$, we have a boundary superfluid described by two canonically conjugate fields, a $2\pi m$ periodic scalar $\varphi$ (with residual $\bZ_m$ gauge symmetry $\varphi \mapsto \varphi + 2\pi$) and a $U(1)$ gauge field $\vartheta$. The Lagrangian is
\[\label{3dtheory}\mathcal{L}_{\rm 2d\ boundary} = \frac{1}{2\pi m} d\varphi \wedge d\vartheta + \cdots,\]
where $\cdots$ denotes nonquantized couplings. The symmetries act by
\begin{align}
U(1)_{\rm mat} : \begin{cases} \varphi \mapsto \varphi + q \alpha \\ \vartheta \mapsto \vartheta \end{cases}, & & U(1)_{\rm mag} : \begin{cases} \varphi \mapsto \varphi \\ \vartheta \mapsto \vartheta + \lambda \end{cases}.
\end{align}
Note the magnetic 1-form symmetry of $a$ acts as the electric 1-form symmetry of $\vartheta$. We can think of $\varphi$, $\vartheta$ as particle-vortex dual fields describing a 2d superfluid where $U(1)_{\rm mat}$ is spontaneously broken. The electric 1-form symmetry of $\vartheta$ is equivalent to the winding symmetry of $\varphi$. The mixed anomaly between these symmetries of the superfluid was recently noted in \cite{Delacr_taz_2020} in a different context, where the winding symmetry was emergent.

We can see this anomaly as follows. The kinetic term corresponds to the canonical commutation relations
\[\label{eqnsfconjugate}[\varphi(x),d\vartheta(y)_{ij}] = 2\pi i m \epsilon_{ij} \delta(x-y).\]
From the symmetry actions we thus find the matter current takes the form
\[\label{eqnnaivematcurrent3d}J_{\rm mat,\partial} = \frac{q}{2\pi m} d\vartheta.\]
The minimal coupling to background gauge field is
\[S_{\rm mat} = - \frac{q}{m} \int A_{\rm mat} \frac{d\vartheta}{2\pi}.\]
We see that a magnetic symmetry gauge transformation $\vartheta \mapsto \vartheta + \lambda$ produces an anomalous variation cancelling $\delta S_{\rm ``SPT"}$ in \eqref{eqnanomalousvariationmag}. We can likewise show the same for \eqref{eqnanomalousvariationmat}.

We have thus verified that the boundary theories in eqn. \ref{2dtheory}, \ref{3dtheory} have the correct anomaly structure to qualify as  boundary theories of the $d=2$ and $d=3$ superconductors respectively.

\subsubsection{Mixed Hall response}

Let us comment on another aspect of the boundary anomaly, which relates current non-conservation at the boundary to a mixed Hall response in the bulk. A general relation between this point of view on the anomaly and the gauge invariance point of view used above is outlined in Appendix \ref{appanomrelationships}.

In the presence of background magnetic gauge field $B_{\rm mag}$, the matter currents \eqref{eqnnaivematcurrent2d} and \eqref{eqnnaivematcurrent3d} are not gauge invariant, and should be modified in $d = 2$ to
\begin{align}
\label{eqnluttingercurrent}J_{\rm mat,\partial} = \frac{q}{2\pi m} (d\theta - B_{\rm mag}), & & J_{\rm mag,\partial} = \frac{1}{2\pi m} (d\varphi - q A_{\rm mat}),
\end{align}
and in $d = 3$ to
\begin{align}
\label{eqnssfcurrents}J_{\rm mat,\partial} = \frac{q}{2\pi m} (d\vartheta - B_{\rm mag}), & & J_{\rm mag,\partial} = \frac{1}{2\pi m} (d\varphi - q A_{\rm mat}).
\end{align}
Doing so spoils the conservation law, for example
\[\label{eqnluttingerboundaryanom}dJ_{\rm mat,\partial} = - \frac{q}{2\pi m} dB_{\rm mag}.\]

We can understand this non-conservation in terms of anomaly in-flow as follows. We compute the bulk current by
\[\label{eqnhallconductivity} J_{\rm mat,bulk} = \frac{\delta S_{\rm ``SPT"}}{\delta A_{\rm mat}} = \frac{q}{2\pi m} dB_{\rm mag}.\]
We can think of this as a mixed Hall conductivity between the matter and magnetic symmetries. This is clear in $d = 2$, although \eqref{eqnhallconductivity} also holds in $d = 3$ as a generalized 0-form/1-form mixed Hall conductance, we just need to use the appropriate form degree for $J_{\rm mag}$ and $B_{\rm mag}$.

We see that the combined bulk and boundary current is conserved in the sense that
\[dJ_{\rm mat,\partial} + J_{\rm mat}|_{\partial} = 0.\]
That is, the missing boundary charge is accounted for by the bulk current. We find an analogous in-flow equation is satisfied by the magnetic current.

This mixed Hall conductance in the bulk gives another point of view on the Higgs-SPT phase. Indeed, by the Laughlin argument, the Hall conductance above can be computed by measuring the matter charge of the $2\pi m$ flux. From $S_{\rm ``SPT"}$, we expect this flux to have charge $q$.

We can see this matter charge as follows. The magnetic symmetry background enters the action as (cf. \eqref{eqnmagcscoupl})
\[S_{\rm mag} = -  \int_X \frac{dB_{\rm mag}}{2\pi} \wedge a.\]
Let us take $X = Y \times \bR_{\ge 0}$ and perform a gauge transformation $a \mapsto a + dg$. We get a boundary variation
\[\delta_g S_{\rm mag} = - \int_{Y \times \{0\}} \frac{dB_{\rm mag}}{2\pi} g.\]
Taking $g$ constant, in the $2\pi m$-flux sector this is $- m g$, so any initial state for the path integral in this sector must have gauge charge $m$ to cancel this variation. This is a special feature of magnetic symmetry in abelian gauge theories: we can access the gauge-charged sectors by studying flux sectors for the magnetic symmetry. The ground state in this sector is created by adding one more condensed particle to the zero-flux ground state. This particle carries global charge $q$, giving us the coefficient $q/m$ in $S_{\rm ``SPT"}$.

\subsection{Lattice gauge theory}\label{subseclatticegaugetheory}

In this section, we define a lattice model for the $m = 1$, $q = 1$ Higgs-SPT phase and exhibit its edge modes. Our main result is summarized in Fig.~\ref{fig:LGTphasediagram} for 3+1d. With open boundary conditions that preserve symmetries, a gapless superfluid phase is expected on the 2+1d boundary of a 3+1d Higgs phase, and a gapless Luttinger liquid phase is expected on the 1+1d boundary of a 2+1d Higgs phase. Both are protected by the anomalous $U(1)_{\rm mat} \times U(1)_{\rm mag}$ symmetry. We also discuss explicit magnetic symmetry breaking, for which we expect the edge modes are stable to some finite perturbation strength, and we sketch the phase diagram. This should stimulate future lattice gauge theory calculations to establish the global phase diagram of these surface states.

In order to construct a lattice model for the compact $U(1)$ gauge theory with quantized charges (as opposed to what is sometimes referred to as ``noncompact $U(1)$", meaning $\bR$ gauge theory), we will use the (Hamiltonian) Villain formalism, discussed in \cite{SULEJMANPASIC2019114616}. This will allow us to define a theory where monopoles will be eliminated thereby enforcing the magnetic symmetry; moreover, one can then study the effect of introducing monopoles by explicitly breaking this symmetry. We can think of this as an $\bR$ gauge field $a$ with a $\bZ$ electric 1-form symmetry (present whenever charge is quantized) which we gauge to get $U(1) = \bR/\bZ$. We will then write down a version of this theory with open boundary conditions and show that it naturally realizes the boundary superfluid discussed above.

\begin{figure}
\centering
\begin{tikzpicture}[scale=1.7]
	\foreach \x in {2,3}{
		\draw[-,line width=1.3] (-.7,0-\x) -- (3+.7,0-\x);
        \foreach \blah in {1,2,3,4,5,6,7,8}{
		      \draw[->] (-.7,0-\x) -- (-.7+0.5*\blah,0-\x);
        };
	}
	\foreach \x in {0,1,2,3}{
        \draw[-,line width=1.3] (\x,.7-2) -- (\x,-3.7);
        \foreach \blah in {0,1,2,3}{
		      \draw[->] (\x,-3.7) -- (\x,.7-2-0.4-0.5*\blah);
        };
	};
	\node at (1.5,-4) {$\vdots$};
	\node at (1.5,-1) {$\vdots$};
	\node at (-1,-2.5) {$\dots$};
	\node at (4,-2.5) {$\dots$};
	\foreach \x in {0,1,2,3}{
		\filldraw[orange] (\x,-1.5) circle (2pt);
		\filldraw[blue] (\x,-2) circle (2pt);
		\filldraw[orange] (\x,-2.5) circle (2pt);
		\filldraw[blue] (\x,-3) circle (2pt);
		\filldraw[orange] (\x,-3.5) circle (2pt);
	};
	\foreach \x in {0,1,2,3,4}{
		\filldraw[orange] (\x-0.5,-2) circle (2pt);
		\filldraw[orange] (\x-0.5,-3) circle (2pt);
        \filldraw[green!70!black] (\x-0.5,-3.5) circle (2pt);
        \filldraw[green!70!black] (\x-0.5,-2.5) circle (2pt);
        \filldraw[green!70!black] (\x-0.5,-1.5) circle (2pt);
	};
    \node[green!70!black,above] at (1.5,-2.5) {$(\Theta_p,m_p)$};
    \node[orange,above] at (0.5,-2) {$(a_e,E_e)$};
    \node[blue,below right] at (2,-3) {$(\phi_v,n_v)$};
	%\node at (9,-1.5) {$\leftarrow$ boundary edges and plaquettes};
 \node at (6,-2.5) {\includegraphics[scale=0.8]{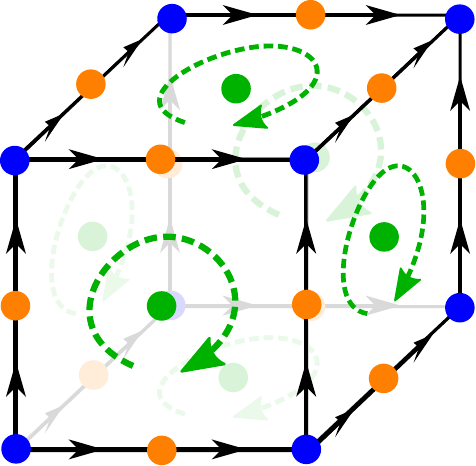}};
 \node at (-1.1,-1.3) {(a)};
 \node at (4.8,-1.3) {(b)};
\end{tikzpicture}
\caption{\textbf{$\boldsymbol{U(1)}$ lattice gauge theory.} In order to study lattice gauge theory with an exact magnetic symmetry, we use the Villain formalism with matter on the vertices (blue), 1-form $\mathbb R$ gauge fields on the links (orange), and 2-form $\mathbb Z$ gauge fields on the plaquettes (green). See Eq.~\eqref{eqnhiggslatticemodel} for the Hamiltonian. (a) Example on the 2D square lattice. The black arrows represent the orientation of the edges. The orientation of the plaquettes can be taken to be clockwise. (b) Example on 3D cubic lattice. The green dashed lines signify the orientation of plaquettes. \label{fig:LGT}}
\end{figure}
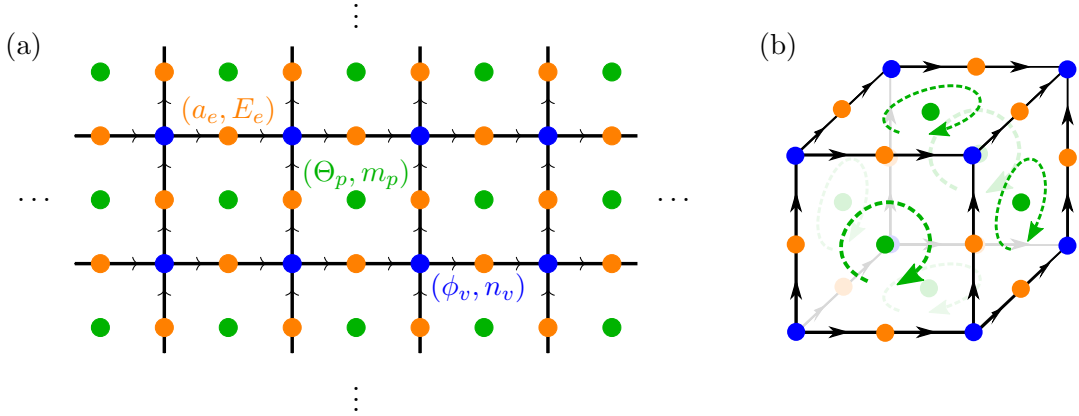

\subsubsection{Hilbert space and model}

The Hilbert space may be defined on any spatial lattice with oriented edges and plaquettes. We will consider degrees of freedom on the vertices, edges and plaquettes, as shown in Fig.~\ref{fig:LGT}. More precisely, the Hilbert space consists of:
\begin{itemize}
\item \textbf{The 1-form $\bR$ gauge field:} a real degree of freedom associated to each edge $e$, with conjugate variables $a_e, E_e$ (both $\bR$ valued) satisfying
\[[a_e,E_e] = -i.\]
\item \textbf{The $U(1)$ matter field:} a rotor degree of freedom associated to each vertex $v$, with conjugate variables $\phi_v$ (taking $\bR/2\pi \bZ$ values) and $n_v$ (taking $\bZ$ values) satisfying
\[[\phi_v,n_v] = -i.\]
\item \textbf{The 2-form $\bZ$ gauge field:} a rotor degree of freedom associated to each plaquette $p$, with conjugate variables $\Theta_p$ (taking $\bR/2\pi \bZ$ values) and $m_p$ (taking $\bZ$ values) satisfying
\[[\Theta_p,m_p] = -i.\]
\end{itemize}

\textbf{Gauge symmetry:} We have an $\bR$ Gauss law at each vertex $v$:
\[\label{eqnrealgausslaw}-n_v = \sum_{e \ni v} (-1)^{ve} E_e =: (\nabla \cdot E)_v,\]
where $(-1)^{ve} = 1$ if the orientation of $e$ points towards $v$, and $-1$ if it points away. This generates the 0-form $\bR$ gauge transformation
\[\phi_v \mapsto \phi_v + g_v \\ 
a_{vv'} \mapsto a_{vv'} + g_{v'} - g_v =: a_{vv'} + (\nabla g)_{vv'},\]
where $vv'$ denotes an edge oriented from $v$ to $v'$.

We also have a $\bZ$ Gauss law at each edge $e$
\[\label{eqnintgausslaw}e^{-2\pi iE_e} = e^{i \sum_{p \ni e} (-1)^{ep} \Theta_p},\]
generating the 1-form $\bZ$ gauge transformation
\[a_e \mapsto a_e + 2\pi \lambda_e \\ 
m_p \mapsto m_p + (\nabla \times \lambda)_p,\]
where
\[(\nabla \times \lambda)_p = \sum_{e \ni p} (-1)^{ep} \lambda_e,\]
where $(-1)^{ep} = 1$ if $e$ is oriented along $\partial p$, or $-1$ if it is oriented against it. As discussed above, we can think of this second gauge symmetry as coming from a $\mathbb Z$ electric symmetry that has been gauged to obtain a $U(1)$ gauge field. Moreover, note that $(\nabla \times a)_p - 2\pi m_p$ is gauge-invariant and can be interpreted as a lattice analogue of the magnetic field.

\textbf{Global symmetry:} The symmetries we wish to preserve are the $U(1)_{\rm mat}$ matter symmetry with generator
\[Q_{\rm mat} = \sum_v n_v\]
and the $U(1)_{\rm mag}$ magnetic symmetry with generators
\[Q_{\rm mag}(\Sigma) = \sum_{p \in \Sigma} m_p,\]
where $\Sigma$ is an arbitrary closed surface made from lattice plaquettes.

\textbf{Model:} A gauge-invariant model with these symmetries (in $d \ge 2$ dimensions) has Hamiltonian
\[\label{eqnhiggslatticemodel}H_{\rm bulk} =-t\sum_e \cos((\nabla \phi)_e - a_e) - \sum_{e} E_e^2 - K \sum_p ((\nabla \times a)_p - 2\pi m_p)^2.\]
In the limit $t \to \infty$, the system is driven into the Higgs phase, and in this limit we have
\[\nabla \phi = a \mod 2\pi.\]
In this fixed point limit, by a 0-form gauge transformation we can eliminate $\phi$ and so reduce $a = 2\pi \lambda$ for some integer $\lambda$, and then by a 1-form gauge transformation eliminate the remaining part of $a$, leaving the trivially gapped model
\[-4\pi^2 K\sum_p m_p^2.\]
Note that the $E^2$ term is projected out because it does not commute with the $\cos(\nabla \phi - a)$ term.

\subsubsection{Boundary theory}

Our analysis here is analogous to that in Part I \cite{partI}, Section 4.4, which discusses the case for $\mathbb Z_2$ gauge theory.

To demonstrate the existence of anomalous edge modes, we will study the model \eqref{eqnhiggslatticemodel} on a semi-infinite geometry with the ``rough" boundary conditions shown for a square lattice in $d = 2$ below:
\begin{center}
\begin{tikzpicture}
	\foreach \x in {2,3}{
		\draw[-] (-.7,0-\x) -- (6+.7,0-\x);
	}
	\foreach \x in {0,1,2,3,4,5,6}{
		\draw[-] (\x,.7-2-.2) -- (\x,-3.7);
	};
	\node at (3,-4.3) {$\bm \vdots$};
	\node at (7.5,-2.5) {$\dots$};
	\node at (-1.5,-2.5) {$\dots$};
	\foreach \x in {0,1,2,3,4,5,6}{
		\filldraw (\x,-1.5) circle (2pt);
		\filldraw (\x,-2) circle (2pt);
		\filldraw (\x,-2.5) circle (2pt);
		\filldraw (\x,-3) circle (2pt);
		\filldraw (\x,-3.5) circle (2pt);
        \filldraw (\x-0.5,-3.5) circle (2pt);
        \filldraw (\x-0.5,-1.5) circle (2pt);
	};
	\foreach \x in {0,1,2,3,4,5,6,7}{
		\filldraw (\x-0.5,-2) circle (2pt);
		\filldraw (\x-0.5,-3) circle (2pt);
        \filldraw (\x-0.5,-2.5) circle (2pt);
	};
	\node at (8.2,-1.7) [align=center] {$\leftarrow$ boundary edges\\ and plaquettes};
	\filldraw[orange] (4,-1.5) circle (2.5pt) node[above] {$a_e$};
	% \filldraw[blue] (5,-1.5) circle (2.5pt) node[right] {$\sigma^z$};
	% \filldraw[blue] (4.5,-2) circle (2.5pt) node[below] {$\sigma^z$};
    \filldraw[green!70!black] (4.5,-1.5) circle (2.5pt) node[above] {$m_p$};
	% \node[orange] at (4.5,-1.6) {$m_p$};
\end{tikzpicture}
\end{center}
The boundary degrees of freedom consist of dangling edges and plaquettes, but no sites. This choice is convenient because every site is unambiguously associated with a 0-form gauge transformation and every edge is unambiguously associated with a 1-form gauge transformation. (As discussed in Part I for the case of discrete gauge theory, one can also derive the edge modes for other boundary conditions if one is careful to modify the boundary Gauss laws to ensure that the symmetries protecting the SPT phase are respected.)

We take the Hamiltonian to consist of all the terms of \eqref{eqnhiggslatticemodel}, where we interpret $(\nabla \times a)_p$ at a dangling plaquette to involve only the three $a$'s which appear there (effectively setting the fourth to zero). Note that we can use the Gauss law to rewrite the matter symmetry as acting purely on the boundary:
\[\label{eqnlatticemattersymmetry}Q_{\rm mat} = \sum_{e \in \partial} E_e,\]
where the sum is over all dangling edges (see \cite{Cobanera_2013} for another discussion of such symmetries).

In the $t \to \infty$ limit, we can use the fixed-point equation $\nabla \phi = a$ mod $2\pi$ and 0-form and 1-form gauge transformations to eliminate all the bulk degrees of freedom. We are left with a boundary theory consisting of dangling edges (which now look like boundary sites $v$) with real degrees of freedom $a_v$, and dangling plaquettes (which now look like boundary edges $e$) with integer degrees of freedom $m_e$. There is no more $\bR$ Gauss law \eqref{eqnrealgausslaw}, but the $\bZ$ Gauss law \eqref{eqnintgausslaw} remains, and becomes
\[\label{eqnboundarygausslaw}e^{-2\pi i E_v} = e^{i \sum_{e \ni v} (-1)^{ve}\Theta_e}.\]
We can think of these two boundary degrees of freedom as a real scalar $a_v$ coupled to a $\bZ$ 1-form gauge field $m_e$ which effectively makes it a periodic scalar!

We are left with the effective model
\[H_{\rm boundary} = - K \sum_e ((\nabla a)_e - 2\pi m_e)^2\]
with symmetry generators
\begin{align}
Q_{\rm mat} = \sum_v E_v, & & Q_{\rm mag}(\gamma) = \sum_{e \in \gamma} m_e,
\end{align}
where the latter sum is over closed curves $\gamma$ along the boundary edges (dangling bulk plaquettes). For the 1d edge of the $d = 2$ Higgs phase, this model actually appeared recently in \cite{mengnati} (cf. Eq. 4.8), where it was demonstrated to describe a $c = 1$ compact boson, where our two $U(1)$ symmetries above (now taking $\gamma$ to be the entire boundary) share the expected mixed anomaly \eqref{eqnU1higgsbackgroundterm} (for $q = 1$). For the 2d boundary of the $d = 3$ Higgs phase, this model describes a superfluid with explicit winding number $Q_{\rm mag}(\gamma)$ conservation, which has the expected mixed anomaly \eqref{eqnU1higgsbackgroundterm} with the spontaneously broken matter symmetry \cite{Delacr_taz_2020}. In either case, the anomaly protects this gapless boundary mode. In particular, while in the fixed-point limit we were able to directly derive the microscopic gapless boundary theory, the anomalous symmetry action means that the edge cannot be gapped out without driving a bulk phase transition.

\begin{figure}
\centering
\begin{tikzpicture}
\node at (-0.5,4.5) {\includegraphics[scale=0.3]{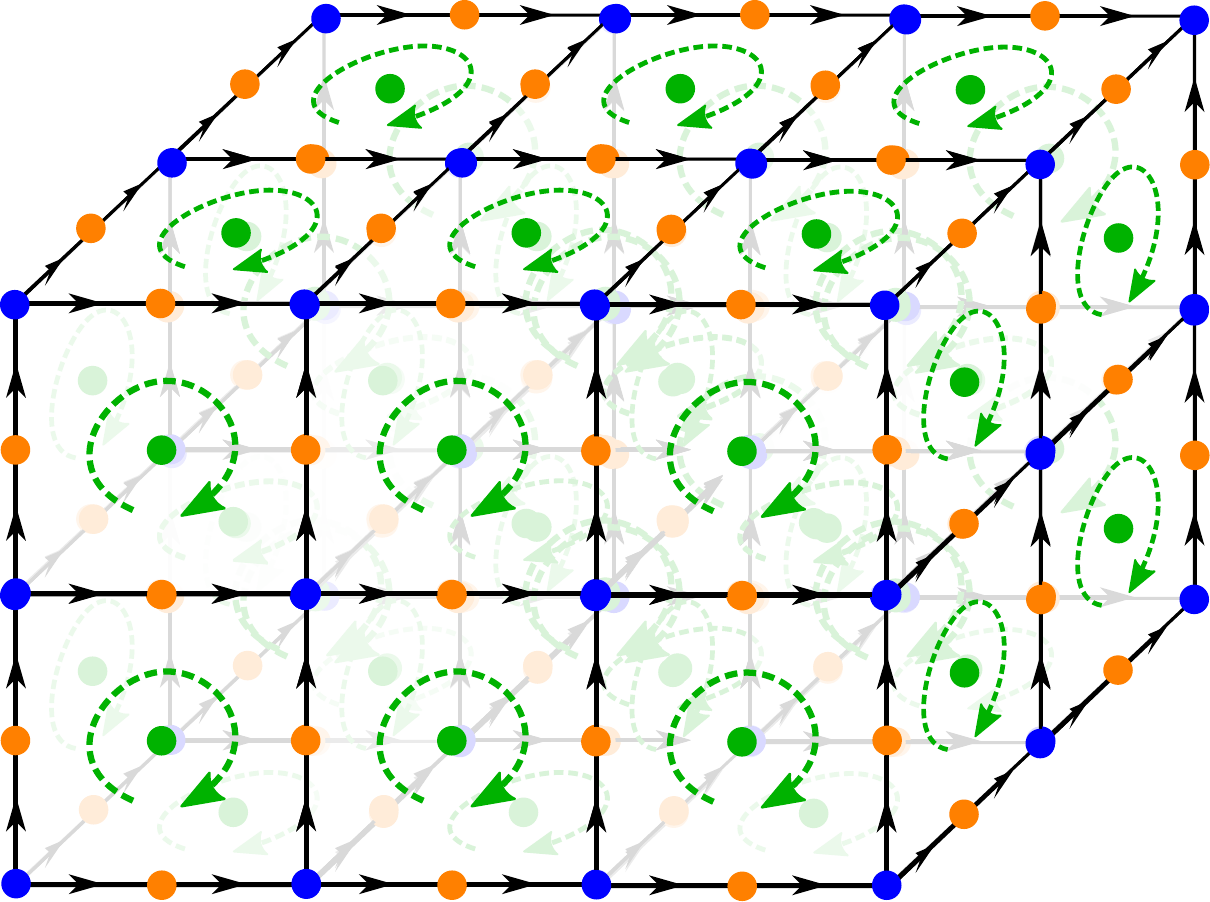}};
\node at (4.5,4.5) {\includegraphics[scale=0.3]{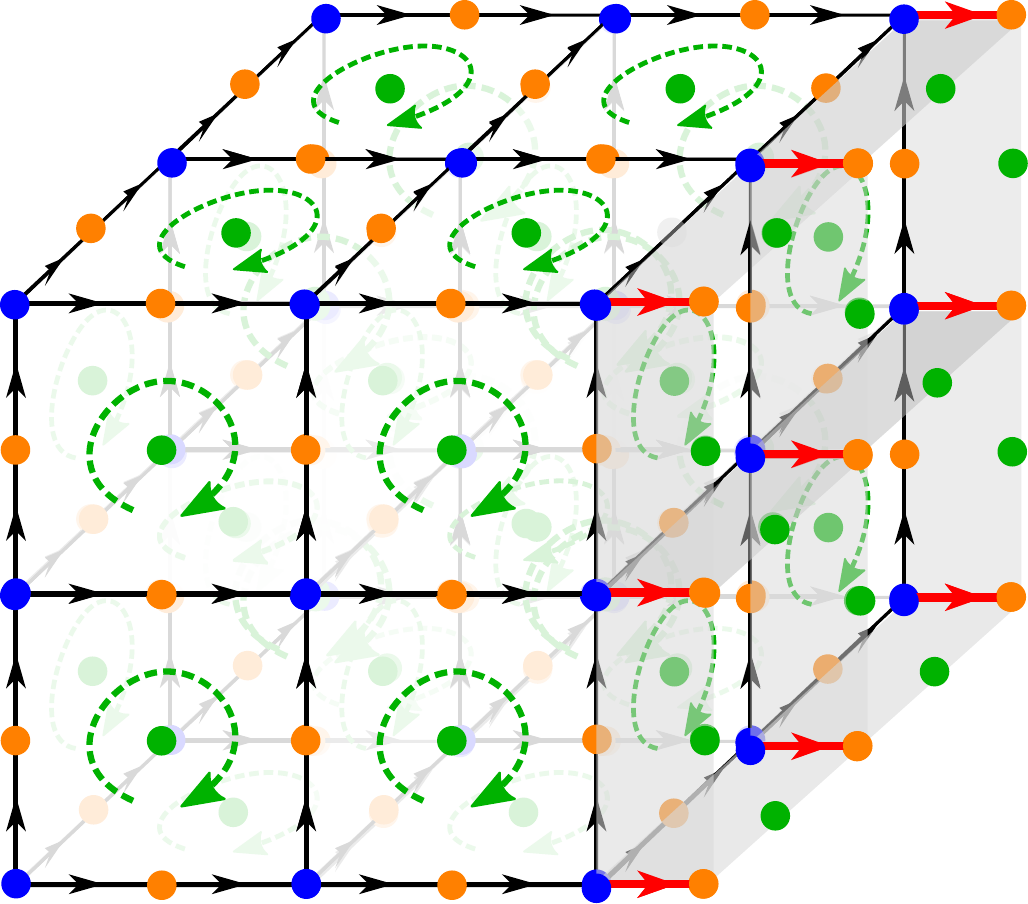}};
\node at (8.5,4.5) {\includegraphics[scale=0.3]{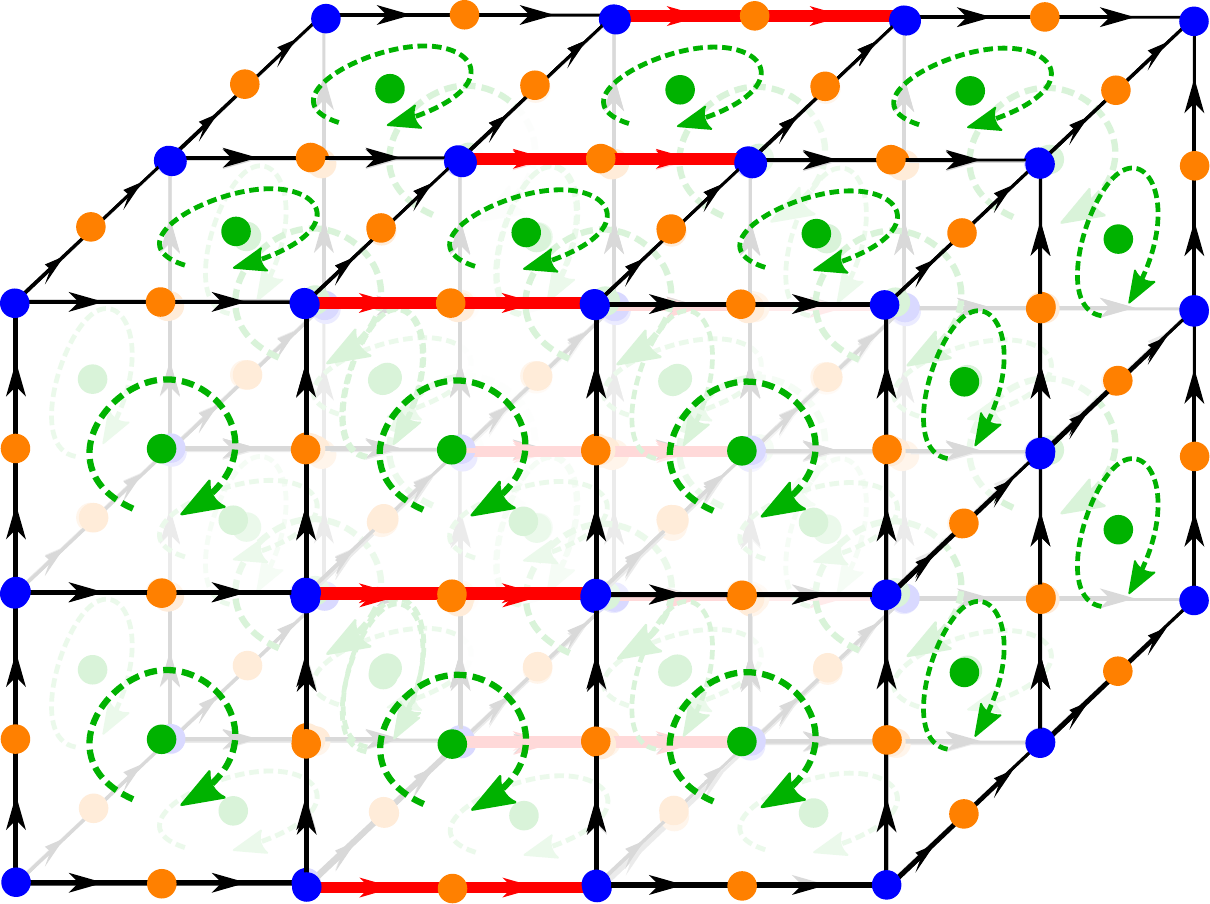}};
\node[right] at (-2,2.3) {(a) periodic b.c.};
\node at (0,0){
    \begin{tikzpicture}
    \filldraw[color=black!10!white] (0,1.8) -- (1.8,1.7) -- (2,0) -- (0,0);
    \draw[-] (0,1.8) -- (1.8,1.7) -- (2,0);
    \draw[->] (0,0) -- (4,0) node[right] {$\lambda$};
    \draw[->] (0,0) -- (0,3.5) node[left] {$t$};
    \node at (0.8,2.6) {Higgs};
    \node at (3,0.8) {confined};
    \node at (1,0.8) {Coulomb};
    \end{tikzpicture}
};
\node[right] at (-3.1+6,2.3) {(b) open (\color{red}rough\color{black}) b.c. or \color{red} SIS \color{black} defect};
\draw[->,red] (5.3,2.6) -- (5.3,3);
\draw[->,red] (7.5,2.6) -- (7.8,3);
\node[right,red] at (7.8,2.8) {\scriptsize $t=0$ on red bonds};
\node at (6,0){
    \begin{tikzpicture}
    \filldraw[blue!20!white] (1.8*0.9,1.8-0.1*0.9) to[out=45,in=-100] (2.8-0.25,3.3) -- (0,3.3) -- (0,1.8);
    \draw[-,dashed,red,line width=1.5] (1.8*0.9,1.8-0.1*0.9) to[out=45,in=-100] (2.8-0.25,3.3);
    \filldraw[color=black!10!white] (0,1.8) -- (1.8,1.7) -- (2,0) -- (0,0);
    \draw[-] (0,1.8) -- (1.8,1.7) -- (2,0);
    \draw[->] (0,0) -- (4,0) node[right] {$\lambda$};
    \draw[->] (0,0) -- (0,3.5) node[left] {$t$};
    \node[blue] at (1.2,3.05) {\small Josephson};
    \node[blue] at (1.2,2.65) {\small effect};
    \node[red] at (3.55-0.15,3.05) {\small boundary};
    \node[red] at (3.3-0.15,3.05-0.4) {\small or defect};
    \node[red] at (3.15-0.15,3.05-0.8) {\small transition};
    \node at (0.8,2.15) {Higgs};
    \node at (3,0.8) {confined};
    \node at (1,0.8) {Coulomb};
    \end{tikzpicture}
};
\end{tikzpicture}
\caption{\textbf{Schematic phase diagram for $\boldsymbol{U(1)}$ lattice gauge theory.} We consider the zero-temperature phase diagram of Eq.~\eqref{eqnhiggslatticemodelwithmono} in 3+1D. (a) In the bulk phase diagram, it has been appreciated that the Higgs regime (i.e., large $t$) and confined regime (i.e., large $\lambda$) are two extremes of a single short-range entangled phase of matter \cite{Fradkin79}. (b) Here we point out that the two regimes are separated by a lower-dimensional quantum phase transition in the case of a system with (symmetry-preserving, i.e., rough) boundaries or a system without boundaries but with a bulk `SIS' defect across which charge cannot move. More precisely, in the absence of monopoles ($\lambda=0$), the Higgs phase is a non-trivial SPT phase involving the 1-form magnetic symmetry as well as the global matter symmetry acting on the boundary or SIS defect. The SPT anomaly stabilizes a 2+1D superfluid on the boundary or defect (blue shaded region). This phenomenology is robust to explicitly breaking the magnetic symmetry, until one encounters a quantum phase transition (red dashed lined). Physically, this region of stability implies the presence of a Josephson effect. See Part I \cite{partI} for the analogous case of $\mathbb Z_2$ gauge theory.}
\label{fig:LGTphasediagram}
\end{figure}
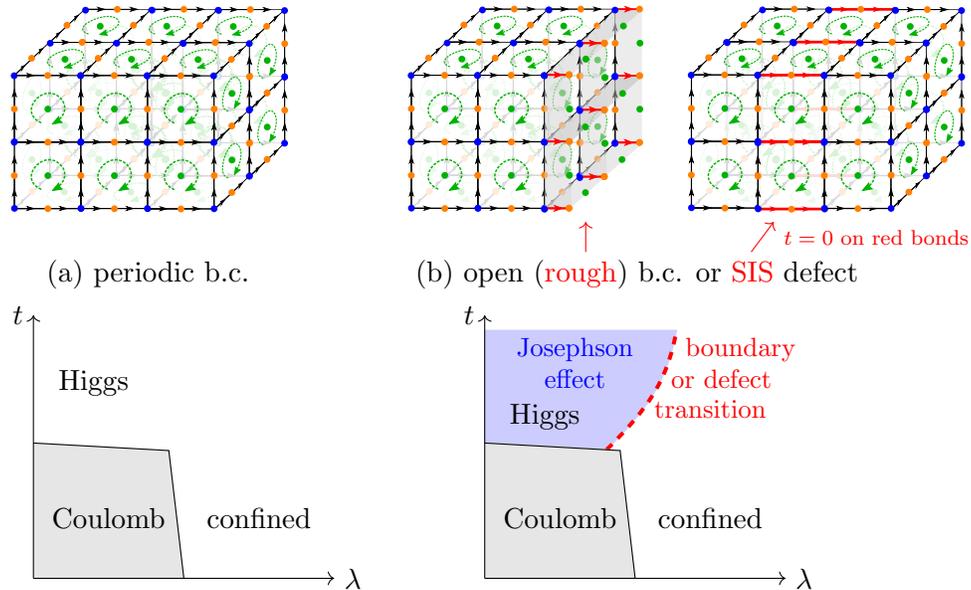

\textbf{Magnetic monopoles:} Thus far we considered a lattice model with an exact charge and magnetic symmetry. The latter physically corresponds to an absence of magnetic monopoles. We can introduce these by adding a term $\cos(\Theta_p)$, which creates a monopole-antimonopole pair on either side of the plaquette $p$ and causes fluctuations for the $m_p$ field (we also set $K=1$ for convenience):
\[\label{eqnhiggslatticemodelwithmono}H =-t\sum_e \cos((\nabla \phi)_e - a_e) - \sum_{e} E_e^2 - \sum_p ((\nabla \times a)_p - 2\pi m_p)^2 - \lambda \sum_p \cos(\Theta_p).\]
For large $\lambda \to \infty$, these monopoles will condense and lead to confinement of the charged particles. It has been appreciated that, ultimately, the Higgs and confined `phases' are two extreme regimes of a single phase of matter \cite{Fradkin79}. Hence, we expect the schematic (bulk) phase diagram shown in Fig.~\ref{fig:LGTphasediagram}(a). Here we focus on the 3+1D case, where the deconfined Coulomb phase is an extended phase, i.e., it is robust to a low density of magnetic monopoles. The latter is due to the magnetic symmetry being a higher form symmetry. For the same reason, the anomalous edge mode of the Higgs phase will also be perturbatively robust! One needs to drive a boundary phase transition to eliminate it, leading to the schematic phase diagram in Fig.~\ref{fig:LGTphasediagram}(b) for open (rough) boundaries.

Note the similarity with the numerical phase diagram obtained in Part I for $\mathbb Z_2$ lattice gauge theory in 2+1D \cite{partI}. We leave a similar numerical study of $U(1)$ gauge theory in 3+1D to future work. We stress that although the precise location of the boundary phase transition is sensitive to microscopic choices, the robust feature is the very existence of an open region in the phase diagram (enclosing the entirety of the Higgs phase with $\lambda=0$) where the 2+1D boundary has symmetry-breaking long-range order. Similarly, we expect an analogous (boundary) phase diagram for the Fradkin-Shenker model for compact $U(1)$ gauge theory \cite{Fradkin79}, which would also be interesting to study numerically in future work.

\subsubsection{Bulk signatures}\label{subsubseclatticejj}

Thus far we focused on the SPT phenomenology at the boundary of the Higgs phase. Here we briefly touch upon bulk signatures in the 3+1D $U(1)$ lattice gauge theory.

\textbf{Bulk SIS defect and Josephson effect:} Rather than considering open boundaries, a similar phenomenology can be observed by considering a 2D bulk defect across which charge is not allowed to tunnel. We call this an SIS defect, as it is akin to `superconducting-insulating-superconducting' junctions. We visually represent it in Fig.~\ref{fig:LGTphasediagram}(b) as a 2D plane of red bonds where we tune the Higgs fluctuations to zero, i.e., $t=0$ in Eq.~\eqref{eqnhiggslatticemodelwithmono}. It introduces a physical $U(1)$ matter symmetry corresponding to the charge on one side of the defect. The non-trivial SPT nature of the Higgs phase implies a mixed anomaly between this matter symmetry and the magnetic symmetry. We will see in fact in our model this $U(1)$ matter is spontaneously broken---similarly to the physical $U(1)$ matter symmetry acting on the boundary links in Eq.~\eqref{eqnlatticemattersymmetry}. Indeed, we can reduce this set-up to the previous analysis by noting that if one takes the fixed-point limit $t \to \infty$ and $\lambda \to 0$ on one side of the defect (e.g., the right-hand side of the SIS defect in Fig.~\ref{fig:LGTphasediagram}(b)), one can gauge-fix that region: the remaining system looks exactly like the geometry with a rough boundary! Similarly, introducing monopoles will lead to the schematic phase diagram in Fig.~\ref{fig:LGTphasediagram}(b), where tuning from the charge-1 Higgs phase to the confined phase eventually leads to a quantum phase transition on the 2+1d bulk defect. In Section \ref{subsec:U1SIS} we explain that this superfluid associated to the SIS defect leads to the Josephson effect when turning on pertubative hopping ($t \neq 0$) across the defect.

\textbf{Bulk SPT phase transition:} While the above is a bulk signature in the sense that it does not require boundaries, it was still associated to a lower-dimensional defect. The SPT nature of the Higgs phase can also results in genuine thermodynamic bulk signatures, such as bulk phase transitions. One minimal scenario for this is having \emph{two} Higgs fields, $\phi_v$ and $\phi'_v$, associated to each vertex $v$ of the lattice. In that case we can write two Higgs terms (both coupled to the same gauge field):
\[\label{eqnhiggslatticemodeltwomatter}
H =-t\sum_e \cos((\nabla \phi)_e - a_e)-t'\sum_e \cos((\nabla \phi')_e - a_e) - \sum_{e} E_e^2 - \sum_p ((\nabla \times a)_p - 2\pi m_p)^2.\]
Now there is a bulk matter symmetry, namely the \emph{relative} $Q_{\rm rel} = \frac{1}{2}\sum_v \left(n_v - n_v' \right)$. Indeed, there exist local, charge-1 and gauge-neutral operators, such as $e^{i\left( \phi_v - \phi_v'\right)}$. Our previous analysis carries over to show that the Higgs phase for $\phi$ (i.e., large $t$) and the Higgs phase for $\phi'$ (i.e., large $t'$) are distinct SPT phases protected by the global $U(1)$ matter symmetry $Q_{\rm rel}$ and the $d-2$ form magnetic symmetry. In particular, this implies that although both Higgs phases seem like trivial short-range entangled states of matter, they cannot be smoothly connected whilst preserving the aforementioned symmetries. We leave the interesting study of its bulk SPT phase transition to future work. We refer to Part I \cite{partI} for a more in-depth discussion of these bulk perspectives, including also the case of a single Higgs fields in a lattice gauge theory where the Gauss law is \emph{emergent}, i.e., energetically enforced.

\begin{figure}
    \centering
    \begin{tikzpicture}
    \node at (0,0) {\includegraphics[scale=0.7]{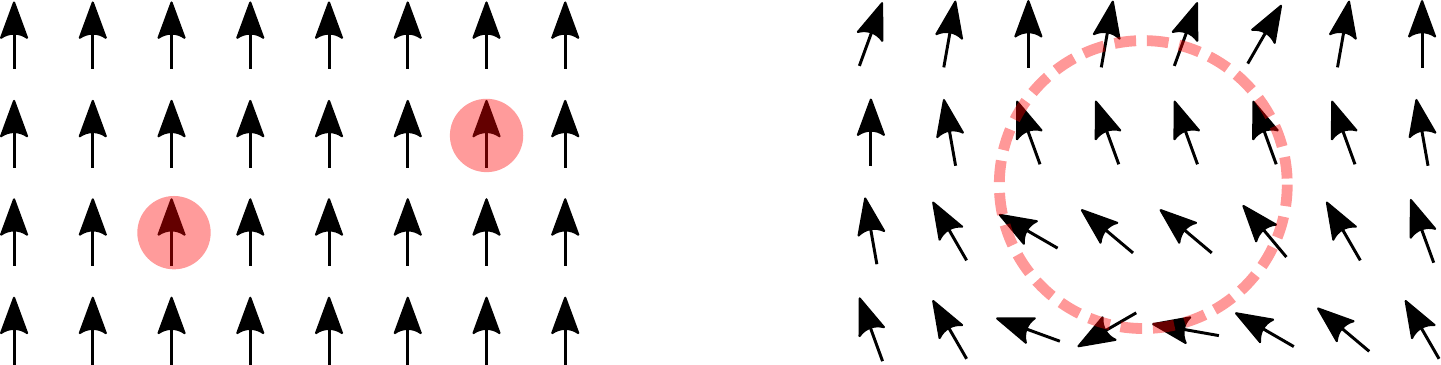}};
    \node at (-5.6,1.1) {(a)};
    \node at (0.5,1.1) {(b)};
    \end{tikzpicture}
    \caption{\textbf{Emergent higher-form symmetries in symmetry-breaking phases.} (a) Spontaneously breaking $\mathbb Z_n$ symmetry leads to a $d$-form symmetry (in $d$ spatial dimensions) which corresponds to the absence of domain walls. The order parameter has a stable value, e.g., its value is the same for the two red dots. (b) Spontaneously breaking $U(1)$ symmetry does not lead to a rigid order parameter since there are low-energy Goldstone modes which continuously deform the long-range order. However, there is an emergent $(d-1)$-form symmetry which characterizes the absence of vortices---which is indeed a quantized defect and hence absent in the ground state. In both cases, this emergent higher-form symmetry shares a mutual anomaly with the global symmetry; these are precisely the anomalies that emerge at the edge of the Higgs SPT phase. Figure taken from Part I \cite{partI}.}
    \label{fig:emergenthigherform2}
\end{figure}

\section{Superconductor Phenomenology from SPT response}\label{secSC}

It is well-known that the Higgs mechanism for the electromagnetic field is a key component of superconductivity. In this section we will show that two important features of superconductor phenomenology, namely Josephson effects (Section \ref{subsec:U1SIS}) and persistent currents (Section \ref{subsecU1higgspt}), can both be derived from the SPT response \eqref{eqnU1higgsbackgroundtermfract} and the boundary anomaly \eqref{eqnanomalousvariationmag} of the Higgs phase.

Another crucial property, the Meissner effect, is simpler. It is a consequence of the fact that the Higgs vacuum is a gapped state with unbroken magnetic symmetry. Intuitively, in any gapped symmetric state, the symmetry charge compressibility vanishes. For the magnetic symmetry, this is the magnetic permeability. Since the leading term in the effective action for the magnetic background field is the generalized Maxwell term $dB_{\rm mag} \wedge \star dB_{\rm mag}$, we find the zero frequency magnetic permeability vanishes as $q^2$, as expected \cite{girvin2019modern}.

\subsection{Josephson Effects \label{subsec:U1SIS}}

One of the hallmarks of superconductivity is the zero-bias current at a superconductor-insulator-superconductor (SIS) junction. Such a junction gives a physical realization of a Higgs-Higgs interface, which we argued above carries a protected interface mode. We will describe how the phase precession in the AC Josephson effect is a feature of the interface anomaly and how this interface mode gives rise to the DC Josephson effect when the matter symmetry is broken by weak tunnelling.

\subsubsection{Weak-Tunneling SIS Junction}\label{subsubsecweaktunnelingjunction}

\begin{figure}[htbp]
    \centering
    \scalebox{0.9}{
    \begin{tikzpicture}
    \node at (0,0) {\includegraphics[width=7cm]{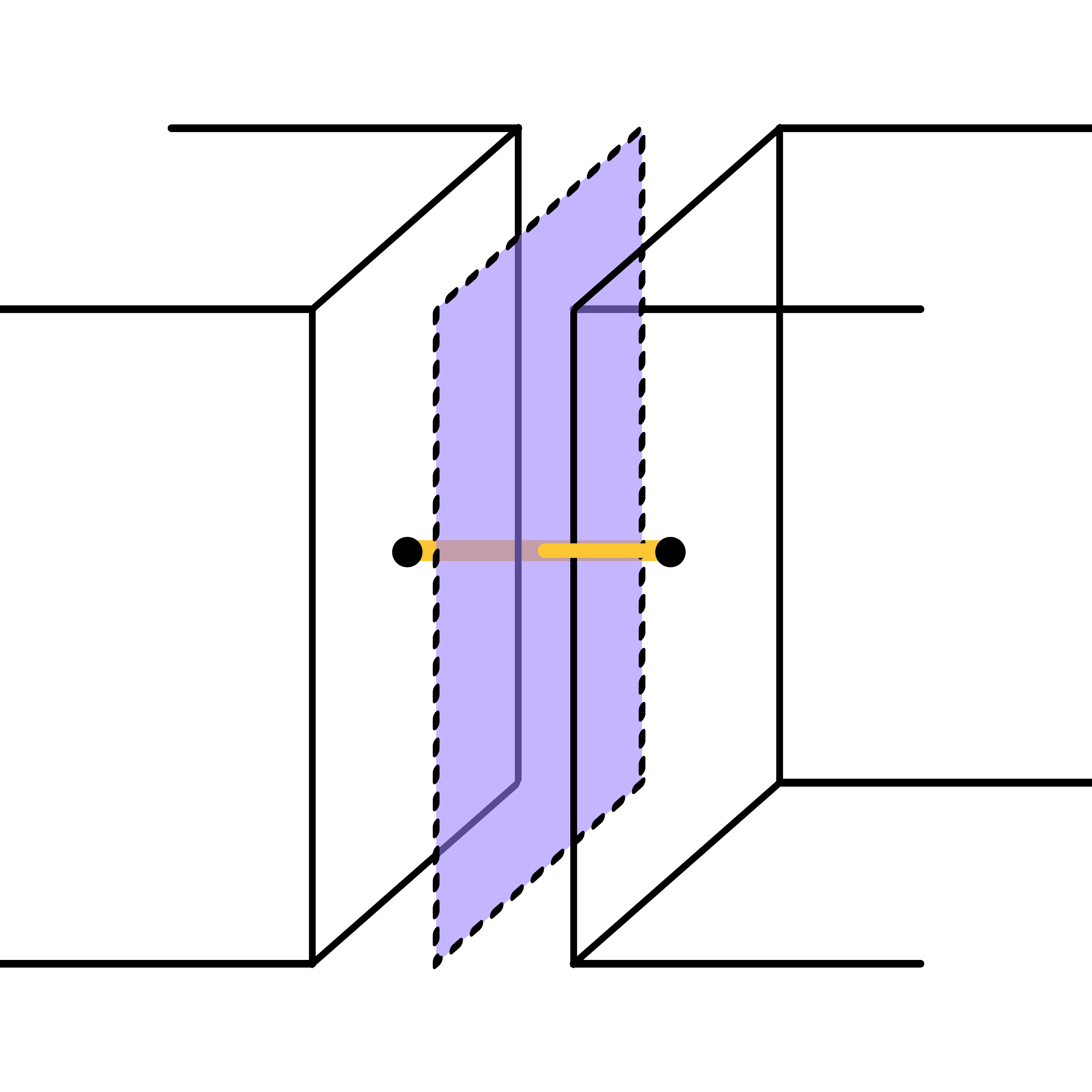}};
    % \node at (-.7,1.1) {$p_{1L}$};
    % \node at (2.7,1.1) {$(x_1,y_1,z_R)$};
    % \node at (-1,-1) {$p_{2L}$};
    % \node at (-2.5,-1) {$(x_2,y_2,z_L)$};
    % \node at (0.7,-1) {$p_{2R}$};
    \node at (-0.5,1) {$E$};
    \node at (5,-1){
        \begin{tikzpicture}
        \draw[->] (5,-0.5) -- (5.8,-0.5) node[right] {$z$};
        \draw[->] (5,-0.5) -- (5.4,0.) node[right,rotate=40] {$x$};
        \draw[->] (5,-0.5) -- (5,0.3) node[above] {$y$};
        \end{tikzpicture}
    };
    \end{tikzpicture}
    }
    \caption{\textbf{Matter symmetry in the Josephson junction.} We consider an SIS junction, where I stands for `Insulator' and S for `Superconductor' or more generally `SPT'. The relative charge $\frac{1}{2}(n_L - n_R)$ between the two halves is equivalent by the Gauss law to the electric flux through the midplane (purple) of the Josephson junction. It is a global symmetry, since it acts on gauge-invariant tunnelling operators (yellow), such as $\cos \varphi_J$. It participates in a mixed anomaly with the magnetic symmetry, revealing the SPT nature of the Higgs phase (Fig. \ref{fig:junctionmagneticsymmetry}), which explains the universal precession of the Josephson current at weak tunnelling.}
    \label{fig:junctionmattersymmetry}
%\end{figure}
%\begin{figure}[!htb]
%\centering
\scalebox{0.9}{
    \begin{tikzpicture}
    \node at (0,0) {\includegraphics[width=7cm]{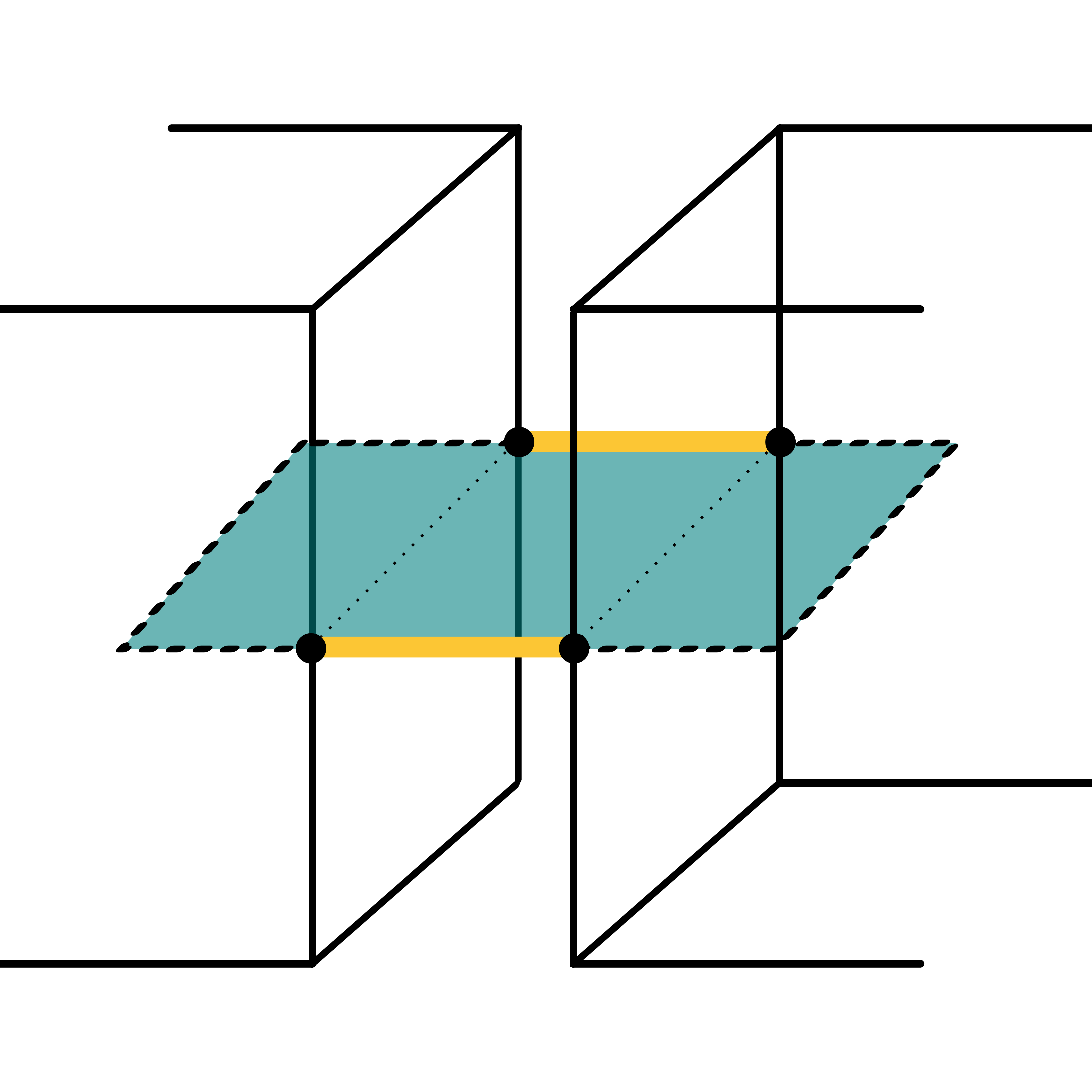}};
    % \node at (-.7,1.1) {$p_{1L}$};
    \node at (2.7,1.1) {$(x_1,y_1,z_R)$};
    % \node at (-1,’.-1) {$p_{2L}$};
    \node at (-2.5,-1) {$(x_2,y_2,z_L)$};
    % \node at (0.7,-1) {$p_{2R}$};
    \node at (-0.5,0) {$B$};
    \node at (5,-1){
        \begin{tikzpicture}
        \draw[->] (5,-0.5) -- (5.8,-0.5) node[right] {$z$};
        \draw[->] (5,-0.5) -- (5.4,0.) node[right,rotate=40] {$x$};
        \draw[->] (5,-0.5) -- (5,0.3) node[above] {$y$};
        \end{tikzpicture}
    };
    \end{tikzpicture}
    }
    \caption{\textbf{Magnetic symmetry in the Josephson junction.} Suppose we measure the magnetic flux through a surface (blue) perpendicular to the Josephson junction. We can express this as the integral of $a$ around the boundary. Along the dotted part of the boundary, we can use the Higgs equations $a = d\varphi_L$ and $a = d\varphi_R$ on either side. We find the flux reduces to $\varphi_J(x_1,y_1) - \varphi_J(x_2,y_2)$ for infinitesimal surfaces. For a larger closed surface, we find the magnetic flux equals the winding number of $\varphi_J$ integrated over the intersection of the surface with the midplane of the junction, and so we identify the magnetic symmetry of $a$ with the winding symmetry of $\varphi_J$. The anomaly can be seen in that these operators are charged under the matter symmetry in Fig. \ref{fig:junctionmattersymmetry}.}
    \label{fig:junctionmagneticsymmetry}
\end{figure}

We will first study superconductor-insulator-superconductor (SIS) junctions between superconductors of the same kind. In the absence of any tunnelling, the relative gauge charge $(n_L - n_R)/2$ across the junction is conserved\footnote{Since the total gauge charge $n_L + n_R = 0$, $n_L-n_R$ is even, so we can normalize by dividing by 2. One could also just take $n_L$ as our junction symmetry if one prefers.}, see Fig. \ref{fig:junctionmattersymmetry}. It generates a \emph{global} symmetry $U(1)_J$, for which the Higgs condensates on the two sides of the SIS junction carry opposite charges. They are therefore distinct SPTs for $U(1)_J \times U(1)_{\rm mag}$, and according to Section \ref{subsechalleffect}, there is an anomalous interface mode on the junction. In Section \ref{subsubseclatticejj} we gave a description of this mode on the lattice.

Let us assume we are in $d = 3$ and the interface mode is described by the superfluid fields $\varphi,\vartheta$ from Section \ref{subsubsectionboundarysuperfluid}. (We propose a physical origin for this interface mode in Section \ref{subsubsec:EMcavity}.) Here $\varphi$ has $U(1)_J$ charge $m$ because $e^{i \varphi/m}$ has 1 $n_L$ and $-1$ $n_R$ gauge charge. The symmetry $U(1)_J$ is spontaneously broken, and the order parameter is
\[\label{eq:LRO}\langle e^{i \varphi} \rangle \neq 0.\]

When we apply a voltage $V$ across the junction, the phase of this order parameter will precess at a universal frequency $eV/h$, a phenomenon known as the AC Josephson effect. This is actually another feature of the interface anomaly, which is that a chemical potential for one symmetry, in this case $U(1)_J$ induces a persistent current for the other, in this case $U(1)_{\rm mag}$.

Indeed, we found in \eqref{eqnssfcurrents} that the magnetic symmetry is identified with the winding symmetry of $\varphi$. This is a 1-form symmetry in $d = 2$, so the current to measure on the 2d interface is the current passing through a given point $x$ on the interface. It is simply
\[\label{eqnjjwindingcurrent}I_{\rm mag}(x) = \frac{1}{2\pi} \langle \partial_t \varphi(x) \rangle = \frac{q}{2\pi m} \mu_J,\]
where on the RHS we have used the persistent current relation derived in Appendix \ref{appsuperfluidanom} to relate this current to the $U(1)_J$ chemical potential. Since $U(1)_J$ is the relative gauge charge, we can identify this chemical potential with the applied voltage by
\[\mu_J = eV.\]
In the present case we have $q = m$, and this becomes the AC Josephson equation (after re-inserting $\hbar$)
\[\label{eqnACjosephson}\langle \partial_t \varphi(x) \rangle = m e V/\hbar.\]

By the same reasoning, we can also apply a chemical potential for $U(1)_{\rm mag}$ and drive a $U(1)_J$ current along the junction. This chemical potential $\mu^{\rm mag}_j$ can be identified with a magnetic field in the junction (parallel to it), according to \eqref{eqnempersistentcurrents}:
\[\mu_j^{\rm mag} = \frac{2\pi}{\mu_0 e} B_j.\]
By the persistent current relation in Appendix \ref{appsuperfluidanom}, we get a $U(1)_J$ current
\[\langle (J_{\rm mat})_i \rangle = - \frac{1}{\mu_0 e} B_i.\]
This is the expected dissipationless current along the junction from Amp\`ere's law. This parallels the derivation of the supercurrent of a superfluid in Ref. \cite{Elsecritdrag}, although here the anomaly is a feature of the SIS interface. See also Section \ref{subsecU1higgspt} for much more discussion of the supercurrent.

The DC Josephson effect occurs in the presence of weak tunnelling across the junction. This explicitly breaks the $U(1)_J$ symmetry, and we can study it in perturbation theory. The most relevant charged operator will be the order parameter $\cos \varphi$ (for simplicity assume reflection symmetry across the junction which eliminates $\sin \varphi$ ). In the limit of small tunnelling, we thus have
\[H_\epsilon = H_0 + \frac{\varepsilon}{A} \int d^2x\,  \cos \varphi (x),\]
where $H_0$ is the $U(1)_{\rm mat} \times U(1)_{\rm mag}$ symmetric  Hamiltonian leading to a superfluid ground state in the junction and $A$ is the area of the junction. We can then measure the charge current \emph{across} the junction as the change in the charge imbalance, and find the usual DC Josephson equation
\[\label{eqncurrentphaserelation}\frac{1}{2}\frac{d(n_L-n_R)}{dt} = i \frac{1}{2}[H_\varepsilon,n_L-n_R] = i \varepsilon \frac{1}{2} [\cos \varphi, n_L-n_R] = \varepsilon \sin \varphi.\]
While these relations hold as operator equations, due to the long range order of $\varphi$ in the junction (\ref{eq:LRO})  we have  replaced the RHS by a uniform  c-number phase which furnishes the DC Josephson relation $I = \epsilon \sin \phi$.\footnote{Note that the tunnelling operator can take a more complicated form, so this equation may be modified \cite{golubov2004current}, but it must be $2\pi$ periodic in $\varphi$. Also, at finite temperature we expect this interface mode to be only quasi-long range ordered, at least up until the bulk critical temperature. However the large stiffness $\rho$  of the electromagnetic field in the cavity of width $\Delta z$, i.e.$\rho = \left (\frac{\hbar}{2e}\right )^2\frac{1}{\mu_0 \Delta z}$ implies that this is a small effect. }

\subsubsection{SIS junctions with matter symmetry and SPT interfaces}\label{subsecedgemodesfldthy}

Above, we used the no tunnelling (or weak tunnelling) condition to define a matter symmetry (or approximate matter symmetry) $\frac{1}{2}(n_L - n_R)$ which shares an anomaly with the magnetic symmetry and protects the interface mode.

Let us now consider a situation with an intrinsic matter symmetry, $U(1)_{\rm mat}$, such as a component $S^z$ of spin, and suppose the Higgs condensates on either side of the junction have different charges $q_L$ and $q_R$ for this $U(1)_{\rm mat}$. We can also consider more generally discrete matter symmetry $G_{\rm mat}$, such as crystal momentum in a pair density-wave, or angular momentum in a $d$-wave superconductor. As long as there is just one component of the condensate, $G_{\rm mat}$ can be considered as a subgroup of $U(1)_{\rm mat}$ and our derivation of $S_{\rm ``SPT"}$ will apply simply treating $A_{\rm mat}$ as a $G_{\rm mat}$ gauge field. Likewise, the boundary theories described in Section \ref{subsechalleffect} will still match the anomaly.

One outcome for the 2d interface in $d = 3$ is that the interface mode spontaneously breaks $G_{\rm mat}$. In this case we should be able to diagnose the relative $G_{\rm mat}$ charge of the condensates by measuring the order parameter $e^{i\varphi}$. This may give a new route to detecting exotic superconductivity. For discrete $G_{\rm mat}$, this symmetry breaking should be stable to finite temperature.

The AC Josephson effect in either case will lead to a precession of the symmetry breaking order parameter in the presence of an applied voltage, as well as a $G_{\rm mat}$ current along the junction in the presence of magnetic field.

On the other hand, the expected DC Josephson current may be zero, since tunnelling operators like $\cos \varphi$ are disallowed by $U(1)_{\rm mat}$. If $G_{\rm mat}$ is discrete, tunnelling operators like $\cos n \varphi$ will be allowed for some $n$, so we should see a modified periodicity in the current-phase relation. If we explicitly break the $G_{\rm mat}$ symmetry, this will generate a perturbation
\[H_\epsilon = H + \frac{\varepsilon}{A} \int d^2x\,  \cos \varphi (x).\]
to leading order (as in the previous section). We thus find a critical current proportional to the strength of the applied field $\varepsilon$.

\subsubsection{Intermezzo: Josephson junctions and EM cavities \label{subsubsec:EMcavity}}

We pause to give a physical description of the interface mode $\varphi$, which above we viewed as a consequence of an anomaly. We study an SIS junction in $d = 3$, which we will consider as an EM cavity, with ``perfect conductor" boundary conditions $m_{L,R} a|_{L,R} = d\phi_{L,R}$ along the surface of the superconductors, where $m_L$ and $m_R$ are the gauge charges of the Higgs fields $\phi_{L,R}$ in the left and right Higgs phases.

With this boundary condition, the electric field is perpendicular to the boundary and the magnetic field is parallel to it. This cavity support a gapless photon mode propagating parallel to the junction, polarized so that the electric field points perpendicular to the junction (this is the so-called TEM mode, see Ch 8 of \cite{jackson1999classical}). This is our SPT interface mode. All other modes of the cavity have a gap at finite width.

The magnetic symmetry is preserved everywhere, since there are still no (dynamical) monopoles in the cavity. However, there is a new symmetry in the cavity since there is also no electric matter there, namely the electric (aka center) 1-form symmetry we discussed in Section \ref{secpureU1}. Since the TEM mode propagates only in the parallel coordinates of the cavity, only one component of this 1-form symmetry is really important at very low energy, namely the conservation of electric flux through the junction. By the Gauss law, this electric flux is equivalent to the charge imbalance $(n_L - n_R)/2$ we studied above (cf. \eqref{eqnlatticemattersymmetry}).

We see that the no-tunneling condition is equivalent to the conservation of this charge, which we can treat as a global 0-form symmetry $U(1)_J$ on the junction (see Fig.~\ref{fig:junctionmattersymmetry}). The electric-magnetic anomaly of \eqref{eqnEManomcomm} appears here as an anomaly on the junction between $U(1)_J$ and the magnetic symmetry $U(1)_{\rm mag}$ (see Fig. \ref{fig:junctionmagneticsymmetry}). This symmetry and the anomaly persists so long as we keep the no-tunnelling condition, even as the cavity gets very narrow.\footnote{The reduction of 1-form symmetries to 0-form symmetries happening here when we compactify the cavity mode is a well-appreciated feature of such symmetries, see for instance \cite{thetatimereversaltemp,Komargodski_2019}.} 

Since the TEM mode has electric field perpendicular to the junction (call this coordinate $z$), we can express it entirely in terms of the $k_z = 0$ Fourier component of the gauge field:
\[\int_{z_L}^{z_R} a_z(x,y,z,t) dz,\]
where the junction extends from $z = z_L$ to $z = z_R$ and $x,y$ are the parallel coordinates, $t$ the time (see Fig.~\ref{fig:junctionmattersymmetry}). This component is not gauge-invariant, but we can combine it with the Higgs condensates $\phi_L$ and $\phi_R$ to define the gauge-invariant 2+1d scalar
\[\label{eqninterfacefield}\varphi_J(x,y,t) = \frac{m_J}{m_L}\phi_L(x,y,z_L,t) - \frac{m_J}{m_R} \phi_R(x,y,z_R,t) + m_J \int_{z_L}^{z_R} a_z(x,y,z,t) dz,\]
where $\varphi_L$ and $\varphi_R$ have gauge charge $m_L$ and $m_R$ respectively, and $m_J = {\rm lcm}(m_L,m_R)$.

By construction, the single gapless mode of $\varphi_J$ is the TEM mode. Observe also that $e^{i\varphi_J(x,y,t)}$ is an open Wilson line stretching from $(x,y,z_L,t)$ to $(x,y,z_R,t)$ and has $U(1)_J$ charge $q_J = q_L \frac{m_J}{m_L} - q_R \frac{m_J}{m_R}$ (for a similar case in $\mathbb Z_2$ gauge theory, see Sec.~3.2.3 in Part I \cite{partI}). Meanwhile, the magnetic symmetry couples to $\frac{1}{m_J} d\varphi_J$, since the parallel component of the magnetic field is proportional to the gradient of $\frac{1}{m_J} \varphi_J$. There are two important cases to consider
\begin{enumerate}
    \item For an interface to a trivial Higgs phase with $m_R = 1$, $q_R = 0$, we find $q_J = q_L$, $m_J = m_L$, and $\varphi_J$ becomes our superfluid field from Section \ref{subsubsectionboundarysuperfluid}.
    \item For an interface between two Higgs phases with the same $m_L = m_R = m$, but with $q_R = 0$, we find $m_J = m$, $q_J = q_L$, and again $\varphi_J$ becomes our superfluid field from Section \ref{subsubsectionboundarysuperfluid}.
\end{enumerate}
Note that in both cases, we have electric symmetry / $U(1)_J$ SSB, with the open Wilson line as the order parameter, ie.
\[\langle \cos \varphi_J \rangle \neq 0,\]
decaying exponentially as a function of the width of the junction.

\subsection{Supercurrents}\label{subsecU1higgspt}

Above we have discussed the anomalous boundary modes arising at a symmetry preserving interface (such as between two Higgs condensates or at a Josephson junction), which must be present to cure the gauge non-invariance of the bulk SPT response. We will see in Section \ref{subsubsecinterfacecounterterm} however that at an interface to a Coulomb phase, a different resolution of the anomaly is allowed, owing to the fact that the Coulomb vacuum spontaneously breaks one of the protecting symmetries. Technically, this means we can write a local counterterm that cancels the gauge variation, without introducing any new gapless interface modes.

The absence of low energy modes besides the photon presents a puzzle for explaining the persistent currents observed in superconductors. What mode carries the current? To pose a sharper puzzle, which was pointed out already by  Bohm \cite{bohm1949}, how do we reconcile the supercurrent with Bloch's theorem \cite{Watanabe_2019}, which forbids any kind of equilibrium current? Below we answer this within the  SPT formulation of superconductors. Although there are no localized anomalous modes in the superconductor-Coulomb problem, we will find there is another anomaly-based mechanism that gives rise to the current.

In Section \ref{subsubsecSPTsupercurrent} we discuss analogous physics in the quantum spin Hall and other SPT phases, which we believe is of independent interest. We will also describe in Section \ref{subsubsecdisorderinginterface} how the SPT edge modes are disordered and removed by coupling to the Coulomb phase, which gives a more microscopic picture of the origin of current-carrying modes.

\subsubsection{Quantum disordered superconductor-Coulomb interface and Thouless pump}\label{subsubsecinterfacecounterterm}

Recall under a magnetic gauge transformation, the topological response of the superconductor has an anomalous boundary variation (cf. \eqref{eqnanomalousvariationmag})
\[B_{\rm mat} \mapsto B_{\rm mat} + d\lambda \\ 
\delta S_{\rm SPT} = q \int_{\partial X} \lambda \wedge \frac{dA_{\rm mat}}{2\pi}.\]
Here $X$ is the superconducting bulk. Previously, in Section \ref{subsechalleffect}, we observed that when $X$ has a boundary, additional boundary fields are required to cancel this variation. We are now interested in a different scenario, an interface between a superconductor and the Coulomb (vacuum) phase. In the Coulomb phase, the magnetic symmetry is spontaneously broken, and we can use it to write a counterterm along the superconductor-Coulomb interface which cancels the anomalous variation of~\eqref{eqnanomalousvariationmag} directly, without the need to invoke additional low-energy modes at the interface.

To construct this counterterm, let $b$ be the Goldstone field for the Coulomb phase. In $d = 2$, $b$ is a $2\pi$ periodic scalar related to the 't Hooft operator $H(x)$ (see Sec.~\ref{subsecemssb}) by
\[e^{i b} = H(x)/|H(x)|.\]
In $d = 3$, $b$ is a $U(1)$ gauge field related to the 't Hooft operator $H(\Gamma)$ by
\[e^{i \int_\Gamma b} = H(\Gamma)/|H(\Gamma)|.\]
$b$ is particle-vortex dual to $a$ in $d = 2$ and electric-magnetic dual to $a$ in $d = 3$. We see from the relations above that under a magnetic transformation, in either case it transforms as
\[b \mapsto b + \lambda.\]

We can thus write the counterterm
\[S_{\rm Thouless} = - q\int_{\partial X} b \wedge \frac{dA_{\rm mat}}{2\pi},\]
which cancels the boundary variation above. In particular, no new interface modes are needed to satisfy the anomaly. We will also explain in Section \ref{subsubsecdisorderinginterface} how the edge modes we derived in Section \ref{subsechalleffect} can be disordered when coupled to the Coulomb phase.

Although there are no extra interface modes, the counterterm $S_{\rm Thouless}$ has interesting physical consequences. In particular, we see that the interfacial matter current receives a contribution
\[J_{\rm mat,\partial} = - \frac{q}{2\pi} db.\]
In $d = 2$ (where the interface is one dimensional $\bR_x$), this leads to the familiar Thouless pump
\[\label{eqnthoulesspump}j^x_{\rm mat,\partial} = \frac{q}{2\pi} \partial_t b,\]
wherein an adiabatic variation in the order parameter $b$ generates a matter current. In $d = 3$ (where the interface is two dimensional $\bR_x \times \bR_y$), we get a new kind of Thouless pump, with the generalized relation
\[\label{eqngeneralizedthoulesspump}j^x_{\rm mat,\partial} = \frac{q}{2\pi} (\partial_t b_y - \partial_y b_t).\]
We will see below that this Thouless pump is responsible for the supercurrent.

\subsubsection{Supercurrent from Thouless pump}

\begin{figure}
    \centering
    \begin{tikzpicture}
    \node at (0,0){\includegraphics[width=5cm]{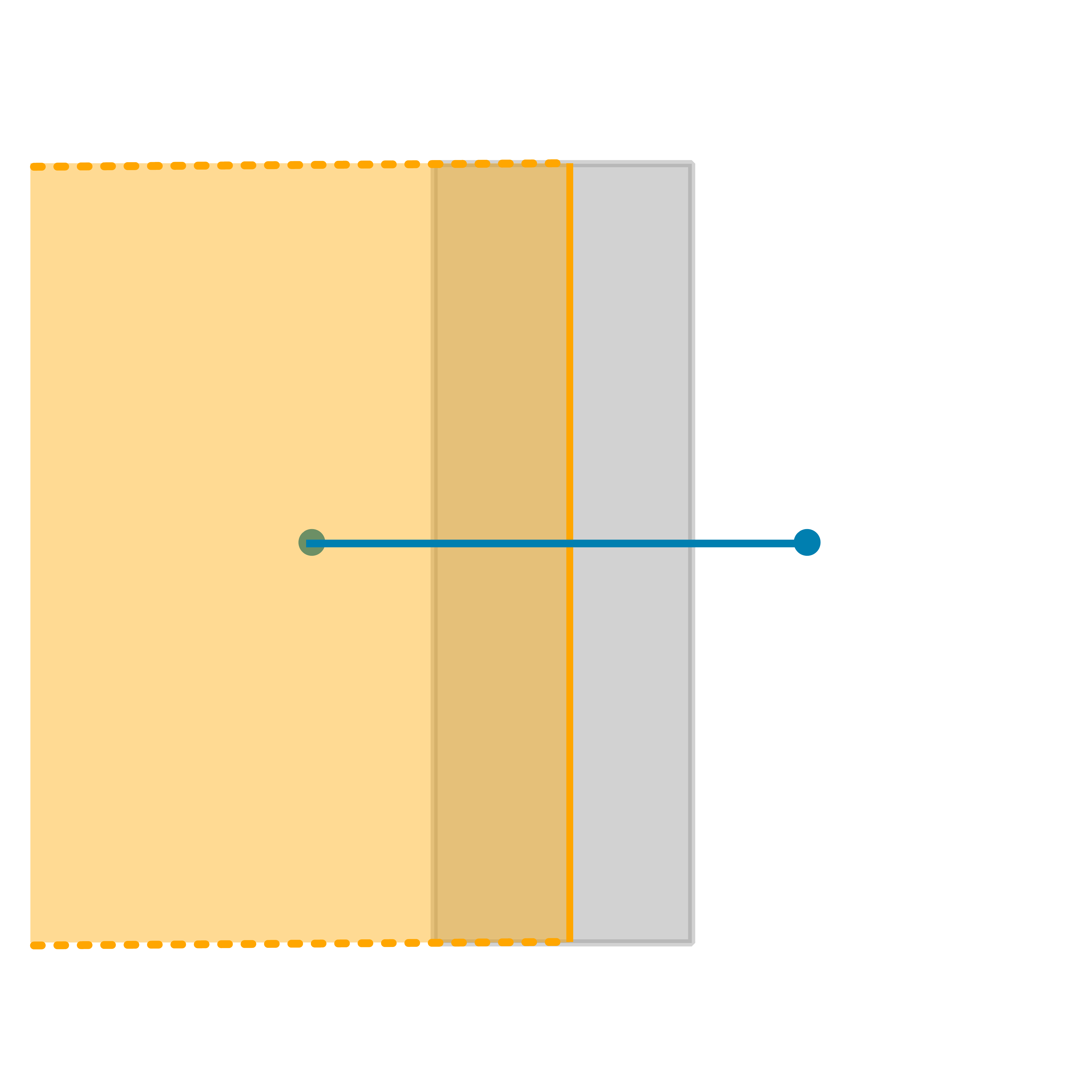}};
        \node at (-3.5,-1){
        \begin{tikzpicture}
        \draw[->] (5,-0.5) -- (5.8,-0.5) node[right] {$x$};
        % \draw[->] (5,-0.5) -- (5.4,0.) node[right,rotate=40] {$x$};
        \draw[->] (5,-0.5) -- (5,0.3) node[above] {$y$};
        \end{tikzpicture}
    };
    \node at (-3,2.5) {(a)};
    \node at (2,0) {$e^{i b(x)}$};
    % \node at (1.7,2) {$e^{i \phi}$};
    \node at (-1.5,1) {$Q_{\rm mag}$};
    \end{tikzpicture}
    \begin{tikzpicture}
    \node at (0,0){\includegraphics[width=5cm]{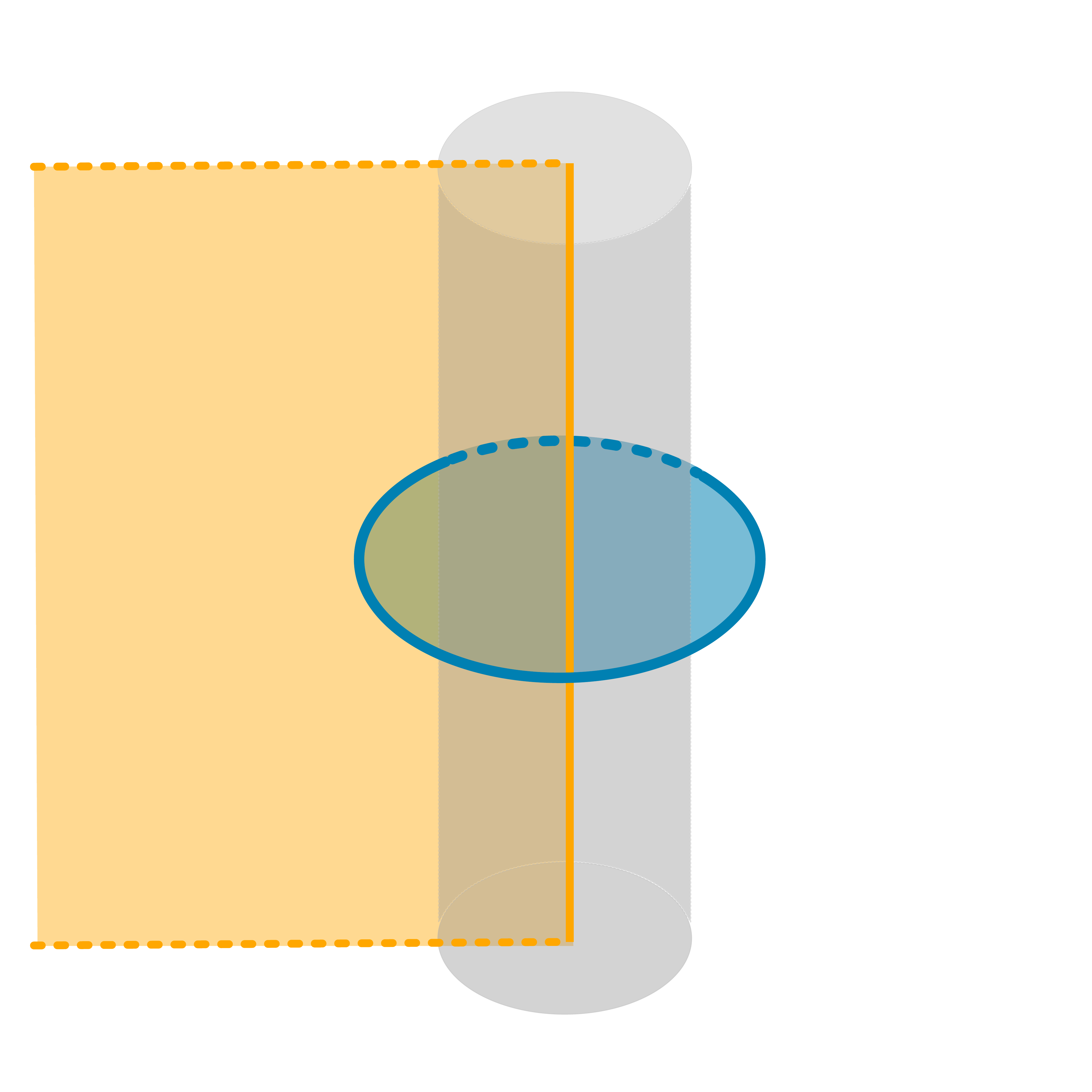}};
    \node at (-3.5,-1){
        \begin{tikzpicture}
        \draw[->] (5,-0.5) -- (5.8,-0.5) node[right] {$x$};
        \draw[->] (5,-0.5) -- (5.4,0.) node[right,rotate=40] {$z$};
        \draw[->] (5,-0.5) -- (5,0.3) node[above] {$y$};
    \end{tikzpicture}
    };
    \node at (-3,2.5) {(b)};
    \node at (2,0) {$e^{i \int_\gamma b}$};
    % \node at (1.7,2) {$e^{i \phi}$};
    \node at (-1.5,1) {$Q_{\rm mag}$};
    \end{tikzpicture}
    \caption{\textbf{Supercurrents as Thouless pumps.} Two geometries for reasoning about the supercurrent in $d = 2$ (panel a) and $d = 3$ (panel b). The superconducting bulk is drawn as a gray strip or cylinder, with Coulomb phase outside. With perfect Meissner effect, magnetic flux cannot tunnel through the superconductor and we thus have an emergent global charge $Q_{\rm mag}$ given by the magnetic flux through the orange surface. This symmetry is spontaneously broken in the Coulomb phase, with order parameter given by a probe monopole, which is a point operator (blue dot) in $d = 2$ and a line operator (blue curve $\gamma$) in $d = 3$. We can identify the magnetic field around the superconductor with a chemical potential for $Q_{\rm mag}$, and find that the phase of this order parameter precesses at a constant rate. Meanwhile, to match the boundary anomaly of the SPT, there is a Thouless pump over the circle of ground states labelled by this order parameter. The constant rate of precession leads to a dissipationless current whose lifetime is set by the decay rate of the emergent symmetry $Q_{\rm mag}$, that is, by the tunelling rate of magnetic flux through the superconductor.}
    \label{fig:supercurrent}
\end{figure}
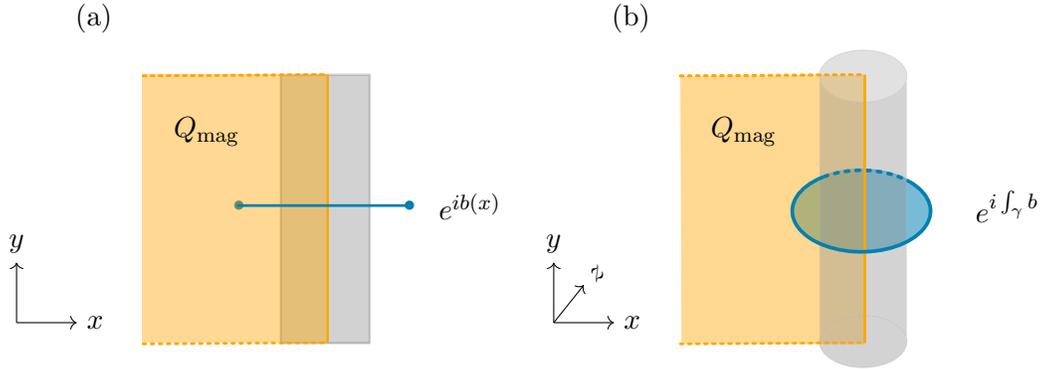

Now we turn our attention to the supercurrent proper. We will consider two analogous geometries in $d = 2$ and $d = 3$ respectively, which are shown in Fig. \ref{fig:supercurrent}. For simplicity we will discuss each case separately.

\subsubsection*{$d = 2$}

Let us first consider two spatial dimensions. We take the superconductor to lie along a strip $[-L,L]_x \times \bR_y$ parallel to the $y$-axis. With perfect Meissner effect, magnetic flux cannot tunnel across the superconductor (quite analogous to the no-tunnelling SIS junction), and as a result the magnetic fluxes through the left and right half-planes are independently conserved. Let $Q_{\rm mag}$ be the difference between these two fluxes. If instead we assume a small tunnelling rate (such as exponentially small in $L$), then the symmetry emerges in the $L \to \infty$ limit.

The symmetry generated by $Q_{\rm mag}$ is spontaneously broken, with order parameter given by the 't Hooft operator $H(x)$ on either side of the superconductor ($x > L$ or $x < - L$ having opposite charges). We write the phase of the order parameter as
\[e^{ib(x)} = H(x)/|H(x)|.\]
According to Section \ref{subsubsecinterfacecounterterm}, as $b$ is adiabatically varied, there is a Thouless matter charge pump with total current
\[I^y_{\rm mat} = \frac{q}{2\pi m} \partial_t b. \label{eqn:Pump}\]

Suppose now we are in a state with a constant magnetic fields $B_L, B_R$, respectively, on either side of the strip, and let $\Delta B = B_L - B_R$. According to the persistent current relation \eqref{eqnempersistentcurrents}, we can identify $\Delta B$ with a chemical potential $\mu^{\rm mag}$ for $Q_{\rm mag}$ by
\[\mu^{\rm mag} = \frac{2\pi}{\mu_0 e} \Delta B.\] 
Since $b(x)$ has magnetic charge 1 (say for $x > L$)
\[\label{eqnbvariation}\langle \partial_t b \rangle = \mu^{\rm mag}.\]
This is because
\[\partial_t b = i[H,b] = i[H_0,b] - i \mu^{\rm mag} [Q_{\rm mag},b] = i[H_0,b] + \mu^{\rm mag},\]
while the state at any particular time is a ground state of $H_0$, so $\langle [H_0,b] \rangle = 0$. In fact, the above equation can be cast as a persistent current for the winding symmetry of $b$, in analogy with \eqref{eqnjjwindingcurrent} and Appendix \ref{appsuperfluidanom}, which in this case is identified with the electric 1-form symmetry of the deconfined Coulomb field.

Combined with the Thouless pump \ref{eqn:Pump}, we have:
\[I^y_{\rm mat} = \frac{q}{\mu_0 e m} \Delta B.\]
This is the supercurrent as predicted from Amp\`ere's law ($\nabla \times B = \mu_0 J$) reduced to $d = 2$.\footnote{Note that for charge $2e$ superconductors, with $U(1)_{\rm mat}$ defined by the relative gauge charge between the superconductor and the environment, $q = m = 2$.} Here instead, we have  derived it from the SPT response and the broken symmetry nature of the Coulomb vacuum. 

One interpretation of this derivation is as follows. The chemical potential difference causes the relative phase across the junction to precess at a constant rate. This can be interpreted as a steady current of phase $b$ vortices flowing along the junction. This would be true for any insulating region between the two Coulomb phases. However, the SPT nature of the superconductor means that the current of vortices also carries a current of electric charge.

\subsubsection*{$d = 3$}

Now we turn to three dimensions. We take the superconductor to lie along a cylinder $\{x^2 + z^2 \le L^2\} \times \bR_y$ parallel to the $y$-axis, shown in Fig. \ref{fig:supercurrent} (b) (a similar argument applies for a torus-shaped superconductor). As in $d = 2$, assuming perfect Meissner effect, we have a magnetic symmetry $Q_{\rm mag}$ given by the magnetic flux through the $x-y$ halfplane with  $z = 0$ and $x < 0$. For finite $L$, $Q_{\rm mag}$ charges will be conserved on exponentially long time scales in $L$.

The symmetry generated by $Q_{\rm mag}$ is spontaneously broken, and the order parameter is a 't Hooft operator $H(\gamma)$ evaluated along a loop $\gamma$ encircling the cylinder as in Fig.~\ref{fig:supercurrent}(b). We write the phase of the order parameter in terms of a 1-form $U(1)$ gauge field $b$, satisfying
\[e^{i \int_\gamma b} = H(\gamma)/|H(\gamma)|.\]
Note $db = 0$ in any Coulomb ground state. In Eqn. \ref{eqngeneralizedthoulesspump} Section \ref{subsubsecinterfacecounterterm}, we showed there is a generalized Thouless matter charge-pump, with total current
\[I^y_{\rm mat} = \frac{q}{2\pi m} \partial_t \int_\gamma b,\]
obtained by integrating \eqref{eqngeneralizedthoulesspump} over a curve encircling the boundary of the cylinder, and using $db = 0$.

Now suppose there is a nonzero, curl-free magnetic field $\vec B$ outside the cylinder. We can use the persistent current relation \eqref{eqnempersistentcurrents} to identify this magnetic field with a chemical potential
\[\mu_{\rm mag}^j = \frac{2\pi}{\mu_0 e} B^j.\]
Since the magnetic symmetry acts as the electric 1-form symmetry of $b$, in the presence of this chemical potential, we have (analogously to \eqref{eqnbvariation})
\[\langle \partial_t b \rangle = \mu^{\rm mag}.\]
Combined with the Thouless pump, we find
\[I^y_{\rm mat} = \frac{q}{\mu_0 e m} \int_\gamma \vec B \cdot d\vec s.\]
This matches the predicted current from Amp\`ere's law.

\subsubsection{Quantum spin Hall supercurrent}
\label{subsubsecSPTsupercurrent} 
To further illustrate the point that the supercurrent depends only on the SPT physics of the Higgs phase, we  now discuss the analogous supercurrent in the 2d quantum spin Hall state.

This state has an identical SPT response to the $m = 1$, $q = 1$ superconductor, but where the roles of the global matter and magnetic symmetries are played by $U(1)_{\rm charge}$ which is simply electric charge (now considered as a global symmetry) and $U(1)_{\rm spin}$ spin symmetries, i.e. spin rotations within an easy-plane \cite{Kane05}
\[S_{\rm QSH} = \frac{1}{2\pi} \int_X A_{\rm charge} \wedge \frac{dA_{\rm spin}}{2\pi}.\]
This action is not gauge-invariant on spacetimes $X$ with boundary (see \eqref{eqnanomalousvariationmag}). Thus, at an interface to a trivial insulator, there must be additional anomalous edge modes to restore gauge invariance. As in Section \ref{subsechalleffect}, one such edge state can be expressed as a $c = 1$ compact boson where $\varphi$ carries unit electric charge and $\theta$ carries unit spin.

However, as we found in Section \ref{subsubsecinterfacecounterterm}, the interface with a $U(1)_{\rm spin}$-breaking ferromagnet does not need to host edge modes, so long as there is a counterterm
\[S_{\rm Thouless} = - \int_{\partial X} b \wedge \frac{dA_{\rm charge}}{2\pi},\]
coupling the phase $b$ of the $U(1)_{\rm spin}$ order parameter to the background $U(1)_{\rm charge}$ gauge field.

This term leads to a Thouless pump as $b$ is varied adiabatically. In particular, if we apply a Zeeman field perpendicular to the easy plane (which we can think of as a chemical potential $\mu^{\rm spin}$) to the ferromagnet, $b$ will begin to precess at a rate $\mu^{\rm spin}$, leading to a ``supercurrent"
\[I^{\rm charge} = (e/\hbar) \mu^{\rm spin},\] where $e$ is the fundamental $A_{\rm charge}$ charge. This is consistent with the spin Hall current we would associate with a chemical potential drop of $\mu^{\rm spin}$ for spin. 

\subsubsection{Quantum disordering the superconductor-Coulomb interface}\label{subsubsecdisorderinginterface}

We conclude this section with some discussion of how the edge modes of Section \ref{subsechalleffect} can become disordered by coupling to the deconfined Coulomb field. This gives a mechanism by which the counterterm $S_{\rm Thouless}$ can appear.

In $d = 2$ the coupling is rather simple. Recall the $c = 1$ boundary theory from Section \ref{subsechalleffect}, described by a pair of dual compact scalars $\varphi, \theta$, with $\varphi$ carrying charge $q$ under $U(1)_{\rm mat}$, and $\theta$ carrying charge 1 under $U(1)_{\rm mag}$. Meanwhile the $U(1)_{\rm mag}$-breaking order parameter $b$ of the Coulomb phase is also a compact scalar carrying charge 1 under $U(1)_{\rm mag}$. Thus we can write the symmetric coupling
\[\label{eqndisordercoupling}-\cos(\theta - b).\]
In a certain range of the marginal parameter, this coupling will be relevant, and completely gap the edge theory. The ground state with this potential has $\theta = b$, and from the form of the current \eqref{eqnluttingercurrent}, we recover the Thouless pump \eqref{eqnthoulesspump}.

The situation in $d = 3$ is a bit more subtle. We consider the superfluid edge of Section \ref{subsechalleffect}, described by the pair of dual fields $\varphi, \vartheta$. Again we have two fields $\vartheta$ and $b$ transforming under the magnetic symmetry. However, we cannot write a term directly coupling them like \eqref{eqndisordercoupling}, since they are gauge fields, and we must respect $\vartheta$ and $b$ gauge symmetry independently. To put it another way, before any coupling, we have independent 1-form symmetries for both $\vartheta$ and $b$, and we want to break this $U(1) \times U(1)$ down to the diagonal subgroup, which is the magnetic symmetry. However, there is no local operator we can introduce that breaks the spurious 1-form symmetry.

Instead, what we are allowed to do is introduce a new field $v$, with $\vartheta$ gauge charge 1 and $b$ charge $-1$. The very existence of this field breaks the spurious 1-form symmetry, but since it has opposite charge under the two fields, the magnetic symmetry remains unbroken. The analog to the operator \eqref{eqndisordercoupling} is a Higgs potential for $v$. Taking the mass of $v$ to infinity restores the extra 1-form symmetry, while condensing $v$ leads to the Higgs constraint
\[\vartheta - b = d {\rm arg}(v),\]
which is the analog of the vacuum condition $\theta = b$ in $d = 2$ above. Combined with the current \eqref{eqnssfcurrents}, we find the generalized Thouless pump \eqref{eqngeneralizedthoulesspump}.

We can give an intuitive description of this mechanism as follows. Recall the boundary superfluid at a Higgs-Higgs interface is stable because the magnetic symmetry acts as the winding symmetry of the order parameter $\varphi$, preventing the fluctuation of vortices. In particle-vortex duality, the vortices are re-interpreted as particles charged under $\vartheta$, and the magnetic symmetry excludes such particles from the theory. However, when there is a neighboring Coulomb phase, we can have the fluctuation of vortices bound to magnetic flux lines without violating the magnetic symmetry. This bound state may be interpreted in the dual frame as the particle $v$, which has opposite charge under $\vartheta$ and $b$.

\section{Outlook}

We have seen that in $U(1)$ gauge theory with magnetic symmetry, the Higgs phase gives rise to SPT physics characterized by global charges associated with the condensing Higgs field.  This identifies the ``gauge symmetry breaking" of the Higgs phase with the ``hidden symmetry breaking" of SPTs, and clarifies the meaning of the former by emphasizing the importance of the magnetic symmetry. This generalizes the findings of Part I which discussed the case of discrete gauge theories \cite{partI}.

We have applied this SPT point of view to superconductors, which are Higgs phases of the electromagnetic field. Here the aforementioned magnetic symmetry is an apparent property of our universe, being physically equivalent to the absence of magnetic monopoles. This viewpoint furnishes an alternate description of SIS Josephson effects in terms of an anomalous interface mode, and  of  surface supercurrents at the interface to the Coulomb vacuum in terms of a Thouless pump. It is interesting to note how these diverse routes to dissipation free current carrying states - superconductors, superfluids, quantum spin hall and Thouless pumps - all appear  within this perspective. Moreover, it sheds light on the fate of superconducting phemonena upon introducing monopoles, where we found that, e.g., the Josephson effect is stable up to a threshold density of monopoles.

Moreover, we have found that global quantum numbers of the Cooper pair condensate distinguish different SPTs (really SETs for physical charge $2e$ superconductors, for which we have also clarified the statistics of the anyons). For instance, the distinction between d-wave and s-wave superconductors (assumed to be both fully gapped for ease of comparison) appeals to rotation symmetry, although both preserve it. In the framework of the present paper they simply correspond to distinct SPT phases protected by the same symmetries, and the distinct transformation properties of their Landau-Ginzburg field under rotation  are analogous to distinct string order parameters in 1D SPTs.  For internal global symmetries or those that are not explicitly broken at an edge, the quantum numbers of the condensate can be probed by studying edge modes and Josephson currents at an interface between two different superconductors. This raises the possibility of using interfaces with a known $s$-wave (``trivial") superconductor to probe the order parameter of a superconductor of unknown type, which is not usually accessible as an observable, since it is not gauge invariant. 
Several results of this kind are well-known, such as the interface between a pair density wave (an SET for $U(1)_{\rm mag}$ and lattice translation) and a trivial superconductor hosting a charge density wave (a translation symmetry breaking state) \cite{Dai_2018}. Such results are usually based on the assumption that both Higgs condensates remain ordered at the interface. The anomaly point of view makes no such dynamical assumption, and may be useful to understand situations with strong quantum or thermal fluctuations, as well as other possible interface states.

Although superconductors, which are coupled to a dynamical gauge field, and neutral superfluids are often discussed in parallel, their  are in fact very different, starting with the  gapped versus gapless nature of their bulk spectrum. Here  we find  that the symmetry preserving boundary of  a $d$ dimensional superconductor (as in an SIS junction) in fact has the same anomaly as a superfluid in $d-1$ dimensions\cite{Delacr_taz_2020,Elsecritdrag}, The only distinction being that some of the emergent symmetries of a superfluid descend from exact symmetries in the superconductor boundary. This establishes a nontrivial link between these distinct phenomena.

We have also argued that the interface between a superconductor and the Coulomb phase  need not host  gapless boundary modes, and we have described how the boundary Luttinger and superfluid modes of Section \ref{subsechalleffect} may be disordered by coupling to the deconfined electromagnetic field, leading to a Thouless pump. However, these couplings may be irrelevant, which leads to the question of how to realize other superconductor-Coulomb interface states, and to study what physical features they have.

Since many of the important features of superconductors can be derived as consequences of the SPT response, these results shed new light about SPT phases with the same topological response, but for different symmetries. For example the quantum spin Hall insulator shares an SPT class with the 2+1d superconductor (with 2d electromagnetism) and we can analyze its interface with an XY ferromagnet (a $U(1)_{\rm spin-z}$ symmetry breaking phase) in a way analogous to the superconductor embedded in the usual (Coulomb) vacuum. We find there is no edge mode, but there is a quantized Thouless charge pump which can be controlled by the ferromagnetic order parameter. Applying a perpendicular magnetic field will cause it to precess and drive a protected current at the QSH-FM interface. Similar proposals recently appeared in \cite{topologicaltransistor,topologicaltransistor2} who called such a device a ``topological transistor". It would be very interesting to explore these features further, studying interfaces between different SPT and SSB phases and also looking for the analogs of  Josephson currents at junctions between SPT and trivial phases.

Another source of gauge symmetries in condensed matter physics are parton theories. Here, the fundamental degrees of freedom are expressed in terms of a gauge field coupled to matter and are widely used to describe quantum spin liquids and fractional quantum Hall states \cite{Wen_book}. It would be interesting to understand under what conditions a physical magnetic symmetry emerges, and then consider Higgs-SPT physics in these systems. We expect that the consequences of the Higgs-SPT we studied, such as supercurrents and Josephson effects, are robust to a finite degree of explicit breaking of the higher form symmetry, similarly to the edge modes studied in Section \ref{subseclatticegaugetheory} and in Part I \cite{partI}. Therefore it is likely to be sufficient for the systems of interest to have an \emph{approximate} magnetic symmetry.

It would also be interesting to generalize our study to finite temperature. SPTs that require a 0-form symmetry are generally expected to be unstable to any finite temperature~\cite{Roberts_2017}. On the other hand, the presence of an exact gauge constraint (which, in our language, corresponds to an SPT stabilizer) might change this picture, and indeed, superconductors are of course stable to finite $T$, so further study is needed.

Finally, our SPT perspective also led to non-trivial predictions for the boundary phase diagram and the fate of the Josephson effect in $U(1)$ lattice gauge theories. In particular, we predict that tuning between a charge-1 Higgs phase and the confined phase will lead to a quantum phase transition in the insulating junction. It would be fascinating for future numerical studies to investigate this, and to determine the precise critical lines.

\section*{Acknowledgements}

We thank Ehud Altman, Erez Berg, Daniel Harlow, Abijith Krishnan, Ho Tat Lam, Leon Liu, Nat Tantivasadakarn, Juven Wang, Steven Simon, T. Senthil, N. Seiberg, Kaixiang Su, Max Metlitski, Stephen Shenker, Cenke Xu and Carolyn Zhang for insightful discussions. We are grateful to Umberto Borla and Sergej Moroz for their collaboration on Part I of this work \cite{partI}.
RV is supported by the Harvard Quantum Initiative Postdoctoral Fellowship in Science and Engineering.
RV and AV are supported by the Simons Collaboration on Ultra-Quantum Matter, which is a grant from the Simons Foundation (651440, A.V.). RT is supported in part by the National Science Foundation under Grant No. NSF PHY-1748958. T.R. is supported in part by the Stanford Q-Farm Bloch Postdoctoral Fellowship in Quantum Science and Engineering.

\bibliography{main.bib}

\appendix

\section{Poincar\'e Duality}\label{apppoincareduality}

We will often be interested in the action of the symmetry generator on charged objects. While  generators of the p-form symmetry are obtained by integrating the conserved currents over closed $d-p$ dimensional manifolds $\Sigma$, the generalized charged objects are defined on $p+1$ dimensional closed manifolds $\Gamma$ (e.g. world lines of $p=0$ symmetry charges). We will be interested in the intersection of this pair of  manifolds, which intersect at points given that their dimensions sum to $d+1$. This signed intersection number can be conveniently manipulated by associating a closed form $\delta_\Gamma$ of dimension $d-p$ with each $p+1$ dimensional closed surface $\Gamma$, such that  the intersection number of $\Gamma$ with the closed $d-p$ dimensional manifold $\Sigma$ is given by the integral $\int_\Sigma \delta_\Gamma$. This is elaborated in greater generality below.

Recall for an $n$-manifold $X$, the cohomology group $H^k(X,\bR)$ is isomorphic to the real vector space of closed $k$-forms $\alpha$ (meaning $d\alpha = 0$) modulo exact $k$-forms $d\beta$ (note $d^2 = 0$). Meanwhile, the homology group $H_k(X,\bR)$ is isomorphic to the real vector space generated by closed sub-$k$-manifolds $Y$ (meaning $\partial Y = 0$) modulo boundaries $\partial Z$ of sub-$k+1$-manifolds $Z$ (note $\partial^2 = 0$).

Given a cohomology class $[\alpha] \in H^k(X,\bR)$ and a class $[Y] \in H_k(X,\bR)$, the integral
\[\int_Y \alpha \in \bR\]
is well defined. Indeed, if we replace $\alpha$ with $\alpha + d\beta$, which has the same class, $[\alpha] = [\alpha + d\beta]$,
\[\int_Y \alpha + d\beta = \int_Y \alpha + \int_{\partial Y} \beta = \int_Y \alpha\]
by Stokes' theorem, and since $\partial Y = 0$. Likewise we can shift $Y$ by $\partial Z$:
\[\int_{Y + \partial Z} \alpha = \int_Y \alpha + \int_{\partial Z} \alpha = \int_Y \alpha + \int_Z d\alpha = \int_Y \alpha,\]
using $d\alpha = 0$.

Moreover, the integral pairing
\[H^k(X,\bR) \times H_k(X,\bR) \to \bR \\ ([\alpha],[Y]) \mapsto \int_Y \alpha\]
induces an isomorphism between cohomology classes of $k$-forms and linear functions on homology classes of sub-$k$-manifolds:
\[\label{eqndualiso}H^k(X,\bR) \cong {\rm Hom}(H_k(X,\bR),\bR).\]
We will not show this, although a proof can be found in \cite{bott1982differential}, Chapter 1, Section 5.

Meanwhile, in an oriented, closed $n$-manifold, we can also define the intersection pairing
\[H_k(X,\bR) \times H_{n-k}(X,\bR) \to \bR \\ 
([Y],[Z]) \mapsto \#(Y \cap Z),\]
where $\#(Y \cap Z)$ indicates the number of points of intersection of $Y$ and $Z$ (possibly after perturbing $Y$ and $Z$ to be transverse) counted with signs according to the relative orientation of $TY \oplus TZ$ and $TX$ at those points. For any given $Z$, this induces a linear map $H_k(X,\bR) \to \bR$. Combined with \eqref{eqndualiso}, this means we can associate to $Z$ a (cohomology class of) closed $k$-form $\delta_Z$, called the Poincar\'e dual of $Z$, satisfying
\[\int_Y \delta_Z = \#(Y \cap Z)\]
for any closed sub-$k$-manifold $Y$. This defines $\delta_Z$ up to exact forms. We can think of $\delta_Z$ as a generalization of the Dirac delta distribution, although note that $\delta_Z$ is smooth.

For example, suppose we have a square torus with 1-periodic coordinates $x$ and $y$. $H_1(X,\bR) = \bR^2$ is generated by the $x = 0$ and $y = 0$ circles, and $H^1(X,\bR) = \bR^2$ is generated by $dx$ and $dy$. The $x = 0$ circle is Poincar\'e dual to $dx$. For instance,
\[\#(\{x = 0\} \cap \{y = 0\}) = \int_{y = 0} dx = 1.\]

\section{Covariant currents}\label{appcurrents}

\subsection{Currents}

In quantum mechanics, charge is carried by discrete entities, such as electrons. To measure the charge in a region is to \emph{count} the number of charges inside. It is customary to define a charge density by dividing this number by the volume of the region. However, this is not a natural way of talking about the distribution of charges, because it depends on a notion of volume, such as coming from a metric. Instead, associated with a collection of points in $d$ dimensions (of space), we can define a Poincar\'e dual $d$-form using only an orientation of the manifold. We can then define a whole theory of densities and even currents using forms.

The covariant current $J$ for a $p$-form symmetry is a $d-p$-form on spacetime. We use the conventions for forms from \cite{nakahara2018geometry}. The components of $J$ in coordinates $x^\mu$ are given by
\[J = \frac{1}{(d-p)!} J_{\mu_1 \cdots \mu_{d-p}} dx^{\mu_1} \wedge \cdots \wedge dx^{\mu_{d-p}}.\]
The usual charge density can be reconstructed from $J$ via
\[\rho(x)^I = \frac{1}{(d-p)!} \epsilon^{IL} J(x)_{L}\]
which carries $p$ anti-symmetric vector indices. Here capital indices indicate spatial multi-indices, and $\epsilon^{IL}$ is the Levi-Civita symbol. There is also the inverse relation
\[J(x)_L = \frac{1}{p!}\epsilon_{LI} \rho(x)^I,\]
using the inverse symbol $\epsilon_{LI}$, which satisfies
\[\frac{1}{p!}\epsilon_{KI} \epsilon^{IL} = \delta_K^L,\]
if $I$ is a $p$-index. Meanwhile the current density can be reconstructed via
\[j^J = -\frac{1}{(d-p-1)!} \epsilon^{JK} J_{0K},\]
which carries $p+1$ anti-symmetric vector indices.

The physical meaning of $J$ is thus that the integral of $J$ over a spatial $d-p$-cycle $W$
\[\int_W J\]
measures the total charge associated with $W$. Meanwhile, for a spatial $d-p-1$-cycle $V$,
\[\label{eqnappcurrentthroughV}\int_{V \times \bR_t} J,\]
measures the total current flowing through $V$. Thus we can think of $J$ as Poincar\'e dual to the worldlines (or worldvolumes in the case of higher form symmetries) of charged objects.

\subsection{Conservation laws}

The full conservation law is
\[dJ = 0.\]
In terms of components, the exterior derivative $d$ is
\[(dJ)_{iJ} = (d-p+1)\partial_{[i} J_{J]},\]
where the bracket indicates antisymmetrization of the indices, meaning we sum over the permutations of the indices, weighted by signs, and divide by the size of the permutation group. This can be written more explicitly as
\[(dJ)_{i_0 \cdots i_{d-p+1}} = \sum_{n = 0}^{d-p+1} (-1)^n \partial_{i_n} J_{i_0 \cdots \hat i_n \cdots i_{d-p+1}},\]
where $\hat i_n$ indicates that $i_n$ has been removed from $i_0 \cdots i_{d-p+1}$.

The mixed time-space components of the conservation law reads
\[\partial_0 J_{iK} - (d-p)\partial_{[i} J_{|0|K]} = 0,\]
where this time we only anti-symmetrize over the spatial indices (so zero is always in the first index of $J_{0\cdots}$). In terms of the charge and current densities, this is
\[\label{eqncontravariantcontinuity}\partial_0 \rho^I + (-1)^{p} \partial_k j^{kI} = 0.\]
For $p = 0$, this reduces to the usual continuity equation. For $p > 0$, there are also pure space components of the conservation law, which read
\[\partial_{[i} J_{J]} = 0.\]
In terms of the charge density, this is
\[\partial_j \rho^{jI} = 0.\]
These pure space ``conservation laws" mean that our charges $\int_W J$ only depend on the homology class of the spatial $d-p$-cycle $W$ they are measured along. For non-relativistic systems one can consider higher symmetries which (say) only obey the time-space conservation law.

\subsubsection{Examples}

For the magnetic 1-form symmetry in $d = 3$, we have $J^{\rm mag} = da$,
\[\rho^i_{\rm mag} = \frac{1}{2} \epsilon^{ijk} (da)_{jk} = e B_i \\
j^{kl}_{\rm mag} = -\epsilon^{klm} (da)_{0m} = e \epsilon^{klm} E_m.\]
The mixed time-space conservation law is Faraday's law, and the pure space conservation law is Gauss' law for magnetism.

As another example, the winding symmetry of a periodic scalar $\varphi$ in $d$ dimensions is a $d-1$-form symmetry with an associated 1-form current $d\varphi$. The charge and current densities are
\[\rho^I_{\rm winding} = \epsilon^{Ij} \partial_j \varphi \\
j^K_{\rm winding} = -\epsilon^K \partial_0 \varphi.\]
The conservation law follows from $d^2 = 0$, or a simple calculation
\[\epsilon^{Ij} \partial_0 \partial_j \varphi - (-1)^{d-1} \epsilon^{jI} \partial_j \partial_0 \varphi = \epsilon^{Ij} \partial_0 \partial_j \varphi - \epsilon^{Ij} \partial_0 \partial_j \varphi = 0.\]

\subsection{Noether procedure}

The covariant current can be derived from the Noether procedure by varying the action
\[\label{eqnnoether}\delta_\lambda S = \int_{X^{d+1}} d\lambda \wedge J + O(\lambda^2),\]
where $\lambda$ is a small $p$-form parametrizing a symmetry transformation. The conservation law follows by integration by parts. We note that this current is only defined up to exact forms (see \cite{Kapustin_2020} for a related discussion). Moreover, shifting the current by an exact form does not change the symmetry action on local operators, which is also given by integrating $J$ over a $d-p$-submanifold surrounding the operator (which is a generalization of $\mathcal{O} \mapsto U\mathcal{O} U^\dagger$).

Exact functionals of the fields, which are automatically conserved, may thus be considered to be the generators of gauge symmetries. An example is $d \star da$, which generates the gauge transformations of $U(1)$ gauge theory. Note that currents like the magnetic current $da$ are not quite exact because $a$ is not a global form (but it is exact in $\bR$ gauge theory). This means that conserved currents form a kind of cohomology theory. See also Def 1.5 of \cite{EDS} for a discussion in classical field theory.

\subsection{Boundary currents}

In the presence of a boundary there can be a boundary term in the Noether procedure which defines a boundary current
\[\delta S = \int_{X^{d+1}} d\lambda \wedge J + \int_{\partial X^d} d \lambda \wedge J_\partial + O(\lambda^2).\]
The ambiguity is now
\[J \mapsto J + d\eta \\ J_\partial \mapsto J_\partial - \eta + d \rho.\]
Integrating by parts and then requiring $\delta S = 0$ for global (ie. flat) $\lambda$ we find the bulk+boundary conservation law
\[dJ = 0 \qquad J|_{\partial X} = - dJ_\partial .\]
This allows us to define a current flowing through a spacetime hypersurface $\Sigma^d \subset X^{d+1}$ with $\partial \Sigma^{d-1} \subset \partial X^d$ by
\[J(\Sigma^d) = \int_{\Sigma^d} J + \int_{\partial \Sigma^{d-1}} J_\partial\]
and the resulting operator will be topological (and therefore conserved).

\section{Anomalies}\label{appanomrelationships}

\subsection{Three viewpoints}

Let $J^m$ label a collection of conserved currents, which are $d - p_m$ forms corresponding to a collection of $U(1)$ $p_m$-form symmetries. Let $A_m$ be the $p_m+1$-form gauge field which couples to $J^m$. We will focus on abelian symmetries but the non-abelian group case can be derived from this discussion by reducing to the maximal torus. Here $d$ refers to the space dimension of the anomalous theory.

There are three equivalent characterizations of local anomalies shared by these currents:
\begin{enumerate}
    \item \textbf{Anomaly in-flow:} There is a $d+1$-dimensional topological theory with a (higher) Chern-Simons action
    \[S_{\rm anom} = \int_{Z^{d+2}} \Omega(A),\]
    where $A$ is shorthand for all the background gauge fields and $\Omega(A)$ is a Chern-Simons $d+2$-form, where the anomalous system lives at the boundary $\partial Z = X^{d+1}$, such that the total bulk+boundary current is conserved. In particular the boundary current is not conserved on its own, but we can interpret the missing charge as flowing into the bulk, meaning
    \[\label{eqninflow}J_{\rm bulk}^m|_\partial + dJ^m = 0.\]
    \item \textbf{Anomalous (non-)conservation:} In the presence of background fields, the currents are no longer conserved, but satify a modified conservation law
    \[dJ^m = \alpha^m(A),\]
    for some (higher) Chern-Weil form $\alpha(A)$ (meaning some polynomial of the curvatures of the background gauge fields).
    \item \textbf{Anomalous commutator:} The charge densities generating the $U(1)$ symmetries do not commute, even though the global symmetry operators do. In particular, there is a contact term in the commutator which takes the form
    \[\label{eqnappanomcommut}[\rho^m(x)^{I},\rho^n(y)^{J}] = \frac{i}{2\pi} \epsilon^{lIJK}\ \beta^{mn}(A)_K\ \partial_l \delta(x-y),\]
    for some matrix of Chern-Weil forms $\beta^{mn}(A)_K$.
\end{enumerate}
We will sketch how these three equivalent views on the anomaly are related.

Let us first show $(1) \Leftrightarrow (2)$. The bulk current is entirely a function of the background fields:
\[J_{\rm bulk}^m = \frac{\delta \Omega}{\delta A^m}.\]
$J_{\rm bulk}^m$ defined this way is a Chern-Weil form, that is, it's a polynomial in the curvatures with no bare $A^n$ appearing. We see from the boundary conservation law \eqref{eqninflow}
\[\alpha^m(A) = -J_{\rm bulk}^m|_\partial.\]
One can also reconstruct $\Omega$ from $\alpha^m(A)$ this way.

Now let us show $(2) \Leftrightarrow (3)$. The idea is to consider the (non-)conservation equation
\[dJ^n = \alpha^n(A)\]
in the case that only one background field, say $A^m$, has a nonzero time(-space) component, and all background are constant in time. We may write this as a chemical potential
\[A_{0I}^m = \mu^m_I.\]
Since this is the only source of time dependence in $\alpha^n(A)$, the (non-)conservation equation in such a background simplifies to
\[dJ^n = \alpha^n(A) = d\mu^m \wedge \tilde\beta^{mn}(A) dt,\]
where $\tilde\beta^{mn}(A)$ is a Chern-Weil form. This means if we compute the change in the total charge in this background (over a $d-p$-submanifold $W$), we will find
\[\frac{dQ_W}{dt} = \int_W d\mu^m \wedge \tilde\beta^{mn}(A).\]
This can be also be computed directly from the Hamiltonian in the presence of the chemical potential $\mu^n$, and from the anomalous commutator we will conclude $\tilde\beta = \beta$.

\subsection{Persistent Currents}\label{apppersistcurrents}

In the presence of a local anomaly as above, we can derive a \emph{persistent current} in the presence of a chemical potential, by generalizing arguments in \cite{Elsecritdrag}. This gives yet another characterization of the anomaly.

We will consider the modified energy
\[H = H_0 - \int d^d x \mu^m_I \rho^I_m.\]
We want to show in the ground state of $H$ (and also including spatial holonomies for the background fields $A$) that
\[\langle j_n^I \rangle = \frac{1}{2\pi} (-1)^{p_n p_m + p_n + 1} \epsilon^{IJK} \mu^m_J \beta_K^{mn},\]
where $I,J,K$ are spatial multi-indices. This has \eqref{eqnempersistentcurrents} as a special case, as well as examples we discuss in the following section.

The general argument proceeds as in Section \ref{subsubsecemanomaly}. Starting with the ground state $|0\rangle$ of $H$, we do an infinitesimal unitary rotation associated with $\rho^J_n \eta_J$, where $\eta_J$ is an arbitrary $p_n$-form, to define a state
\[|\epsilon\rangle = \left(1 + i \epsilon \int d^d x \eta_J \rho^J_n \right)|0\rangle.\]
Because $|0\rangle$ minimizes $H$, we have
\[\langle \epsilon | H | \epsilon \rangle - \langle 0 | H | 0 \rangle = O(\epsilon^2).\]
The vanishing of the order $\epsilon$ piece gives the relation
\[i \epsilon \int d^d x \eta_J \langle [H, \rho^J_n] \rangle = 0.\]
Here $\langle \mathcal{O} \rangle$ is shorthand for $\langle 0 | \mathcal{O} | 0 \rangle$. Recall $H = H_0 - \int d^d x \mu^m_I \rho^{I}_m$. These give rise to two terms
\[(1) = i\int d^d x \eta_J \langle [H_0, \rho^J_n(x)] \rangle \\ 
(2) = i\int d^d x d^d y \eta_J(x) \mu_I^m(y) \langle [\rho^I_m(y),\rho_n^J(x)] \rangle\]
which must be equal.

The first term may be simplified using the conservation law \eqref{eqncontravariantcontinuity}, which holds for $H_0$:
\[i[H_0,\rho_n^J] = \partial_0 \rho_n^J = (-1)^{p_n+1} \partial_l j_n^{lJ}.\]
We can then integrate by parts to get
\[(1) = (-1)^{p_n} \int d^d x (\partial_l \eta_J(x)) \langle j^{lJ}_n(x)\rangle.\]
Meanwhile, in the second term, we use the anomalous commutator \eqref{eqnappanomcommut}
\[i[\rho^{I}_m(y),\rho^{J}_n(x)] = -\frac{1}{2\pi} \epsilon^{lIJK} \beta^{mn}_K(A) \partial_l \delta(y-x)\]
and obtain
\[(2) = -\frac{1}{2\pi} \int d^d x d^d y \eta_J(x) \mu_I^m(y) \epsilon^{lIJK} \beta^{mn}_K(A) \partial_l \delta(y-x).\]
Performing one of the integrals we find
\[(2) = -\frac{1}{2\pi}\int d^d x \epsilon^{lIJK} \partial_l\eta_J \mu_I^m \beta^{mn}_K(A).\]
Setting the two equal and extracting the integrand using the fact that $\eta$ is arbitrary, we get
\[\langle j_n^{lJ} \rangle = \frac{1}{2\pi} (-1)^{p_n + 1}\epsilon^{lIJK} \mu_I^m \beta^{mn}_K(A) = \frac{1}{2\pi}  (-1)^{p_n p_m + p_n + 1} \epsilon^{lJ IK} \mu_I^m \beta^{mn}_K\]

% \[\langle j_n^I \rangle = \frac{1}{2\pi} (-1)^{p_n + p_m p_n} \epsilon^{IJK} \mu_J^m \beta_K^{mn}(A).\]

\subsection{Examples}

\subsubsection{1+1d chiral anomaly}

Here we consider the mixed $U(1)_v \times U(1)_a$ (vector and axial) anomaly in 1+1d, which one can think of as the $n_L + n_R$ and $n_L - n_R$ charges of $k$ free fermions. See Chapter 20 of \cite{fradkin2021quantum}.

\begin{enumerate}
    \item The anomaly theory is a mixed Chern-Simons theory
    \[S_{\rm anom} = k \int_{Z^3} A^v \frac{dA^a}{2\pi}\]
    \item The conservation law takes the form
    \[dJ^a = -\frac{k}{2\pi} dA^v.\]
    \item The commutation relation is
    \[[n^v(x),n^a(y)] = \frac{ik}{2\pi} \delta'(x-y).\]
    \item The persistent vector current in the presence of an axial chemical potential takes the form
    \[\langle j^v \rangle = -\frac{k}{2\pi} \mu^a.\]
    One interpretation of $\mu^a$ is that it arises from an applied electric field
    \[\mu^a = - e E\]
    (note the sign, since a positive, or right-pointing electric field turns left-moving charges into right-moving ones). With the above we get the expected relation
    \[\langle j^v \rangle = \frac{ek}{2\pi} E.\]
\end{enumerate}

\subsubsection{3+1d chiral anomaly}

Somewhat analogous to the 1+1d example, we can consider the mixed $U(1)_v \times U(1)_a$ anomaly in 3+1d. We can again think of this as the $n_L + n_R$ and $n_L - n_R$ charges of $k$ free fermions, but where now $L$ and $R$ refer to the helicity of Weyl fermions. There is an axial-gravitational anomaly, but for simplicity we will just consider the mixed axial-vector anomaly.

\begin{enumerate}
    \item The anomaly theory is a higher Chern-Simons theory
    \[S_{\rm anom} = k \int_{Z^3} A^a \left(\frac{dA^v}{2\pi}\right)^2\]
    \item The conservation law takes the form
    \[dJ^a = - \frac{k}{(2\pi)^2} (dA^v)^2.\]
    \item A novelty for this anomaly is that the commutation relation depends on the background field
    \[[n^v(x),n^a(y)] = \frac{i k}{(2\pi)^2}\ \epsilon^{ijk}\  (dA^v)(x)_{ij}\ \partial_k \delta(x-y).\]
    \item There is a persistent vector current
    \[\langle j_v^i \rangle = -\frac{k}{(2\pi)^2} \epsilon^{ijk} (dA^v)_{jk} \mu^a.\]
    This is the chiral magnetic effect: an imbalance in fermion helicities created by the axial chemical potential $\mu^a$ generates a (vector) current along the direction of the magnetic field $\epsilon^{ijk} (dA^v)_{jk} \sim B^i$ \cite{Kharzeev_2014}.
\end{enumerate}

\subsubsection{3+1d EM anomaly}

We consider the anomalous 1-form symmetries of ($k$ copies of) electromagnetism in 3+1d.

\begin{enumerate}
    \item The anomaly theory is a higher Chern-Simons theory
    \[S_{\rm anom} = k \int_{Z^3} B_{\rm elec} \frac{dB_{\rm mag}}{2\pi}\]
    \item The conservation law takes the ABJ form
    \[dJ_{\rm mag} = - \frac{k}{2\pi} dB_{\rm elec}.\]
    \item The commutation relations are
    \[[\rho_{\rm mag}(x)^i,\rho_{\rm elec}(y)^j] = \frac{i k}{2\pi}\ \epsilon^{ijk}\  \partial_k \delta(x-y).\]
    \item The persistent current relation is (cf. \eqref{eqnempersistentcurrents})
    \[-\langle J^{\rm mag}_{0j} \rangle = \frac{e}{2\pi} \langle E_j \rangle = \frac{1}{2\pi}\mu_j^{\rm elec} \qquad -\langle J^{\rm elec}_{0j} \rangle = \frac{1}{\mu_0 e}  \langle B_j \rangle = \frac{1}{2\pi}\mu_j^{\rm mag}.\]
    This has the interpretation that a polarized (or magnetized) vacuum has a frozen in electric (or magnetic) field along the polarization (or magnetization) vector $\mu^j_{\rm elec}$ ($\mu^j_{\rm mag}$ respectively). Note that the chemical potentials must be curl-free for the system to be in equilibrium and for our derivation to apply.
\end{enumerate}

\subsubsection{2+1d Superfluid anomaly}\label{appsuperfluidanom}

We consider the anomaly of a 2+1d superfluid \cite{Delacr_taz_2020,Elsecritdrag}, for the broken symmetry $U(1)^b$ and the winding 1-form symmetry $U(1)^w$, with gauge fields $A^b$ and $B^w$, respectively. For the ordinary superfluid, the level is $k = 1$, but for completeness we work with arbitrary level $k$, which we can think of as arising from a charge $k$ order parameter.

\begin{enumerate}
    \item The anomaly theory is a BF theory
    \[S_{\rm anom} = k \int_{Z^3} A^b \frac{dB^w}{2\pi}\]
    \item The conservation law takes the ABJ form
    \[dJ^b = -\frac{k}{2\pi} dB^w \qquad dJ^w = -\frac{k}{2\pi} dA^b.\]
    \item The commutation relations are
    \[[\rho^b(x),\rho^w(y)^j] = -\frac{i k}{2\pi}\ \epsilon^{jk}\  \partial_k \delta(x-y).\]
    \item There are two persistent currents which are interesting. First, in the presence of a chemical potential $\mu^i_w$ for the winding symmetry, we get a current for the 0-form symmetry
    \[\langle j_b^i \rangle = -\frac{k}{2\pi}\mu^i_w.\]
    We can interpret this as follows. The winding chemical potential promotes a gradient in the phase of order parameter, and this gradient is the current.

    Second, in the presence of a chemical potential $\mu_b$ for the 0-form symmetry, we have a persistent winding current
    \[\langle j_w^{ij} \rangle = \frac{k}{2\pi} \epsilon^{ij} \mu_b.\]
    This has the interpretation that in the presence of $\mu_b$, which acts as a transverse field, the order parameter precesses at a universal rate, which we can think of as a persistent winding current.
\end{enumerate}

\section{Higher form gauge fields}\label{apphigherformproperties}

We collect some useful facts about abelian higher form gauge fields and sketch some proofs of these facts. We will define higher form gauge fields and relate them to so-called ``differential cocycles". A very nice review on differential cohomology can be found in \cite{hopkinssinger}, the book \cite{brylinski2007loop}, or here \cite{arundiffcoh} for a modern take.

The definition of a $U(1)$ $p$-form gauge field is done inductively from a $U(1)$ $p-1$-form gauge field:

\begin{defn}
For $p \ge 0$, a \textbf{$U(1)$ $p$-form gauge field} is a collection of $p$-forms $B_i$ defined on an open cover $U_i$ such that on overlaps $U_i \cap U_j$, $B_i - B_j = dA_{ij}$, where $A_{ij}$ is a $U(1)$ $p-1$-form gauge field defined on $U_i \cap U_j$. For $p = 0$, it is a $U(1)$ scalar function.
\end{defn}

For $p = 1$ this is the ordinary definition of a $U(1)$ gauge field, but it is convenient for inductive proofs to extend the definition to include scalars for $p = 0$. One could also extend to $p = -1$ to include locally constant $\bZ$ valued functions. These objects are also called $\bZ(p+1)$ Deligne $p+1$-cocycles in some references.

\begin{defn}
A \textbf{gauge transformation of a $p$-form $U(1)$ gauge field $B$} is a $p-1$-form gauge field $A$ (defined on the same open cover), and acts by
\[B_i \mapsto B_i + dA_i.\]
\end{defn}

Note that any two open covers have a mutual refinement, and these definitions behave well under refinements, so $p$-form $U(1)$ gauge fields modulo gauge transformations do not actually depend on the choice of open cover. $U(1)$ $p$-form gauge fields modulo gauge transformations define the $p+1$st Deligne cohomology $H^{p+1}(X,\bZ(p+1))$ of the manifold $X$.

Next we outline an isomorphism between $H^{p+1}(X,\bZ(p+1))$ and Cheeger-Simons differential cohomology, which can be considered a kind of Villain formalism for $U(1)$ $p$-form gauge fields:

\begin{thm}
$U(1)$ $p$-form gauge fields, modulo gauge transformations, are equivalent to \emph{differential $p+1$-cocycles}, which are triples
\[(c,h,F) \\ 
c \in Z^{p+1}(X,\bZ) \qquad h \in \Omega^{p}(X,\bR) \qquad F \in \Omega^{p+1}(X,\bR)\]
satisfying the differential cocycle equations
\[dc = 0 \\
dF = 0 \\
dh = F - 2\pi c,\]
modulo gauge transformations
\[c \mapsto c + dn \\
h \mapsto h + df - 2\pi n \\
F \mapsto F\]
for
\[f \in \Omega^{p-1}(X,\bR) \qquad n \in C^{p}(X,\bZ).\]
\end{thm}

\begin{proof}
See Lemma 7.3.4 of \cite{arundiffcoh}.
\end{proof}

One should think of $F$ as the (gauge invariant) field strength, $c$ as the (higher) Chern class, and $h$ as the holonomy. That is, for any $p$-submanifold $\Gamma$, the holonomy is
\[e^{i \int_\Gamma B} = e^{i \int_\Gamma h},\]
where the left hand side must be defined relative to the cover $U_i$ where $B_i$ are defined:
\[\int_\Gamma B = \sum_i \int_{\Gamma \cap U_i} B_i - \sum_{i < j} \int_{\Gamma \cap U_i \cap U_j} dA_{ij}.\]

There are several useful corollaries of the theorem above.

\begin{prop}
A $p$-form $U(1)$ gauge field is determined up to gauge transformations by its holonomies
\[e^{i \int_\Gamma B} = e^{i \int_\Gamma h}.\]
\end{prop}

\begin{proof}
$h$ is determined up to exact pieces by its integrals over $p$-submanifolds $\Gamma$. We know those integrals up to $2\pi$ integers, so we know the gauge equivalence class of $h$. We can also get $F$ from studying the holonomy over small spheres. Finally $c$ comes from studying the winding number of the holonomy around a $p+1$-submanifold.
\end{proof}

\begin{prop}
For any closed, real $p+1$-form $F$ with $2\pi$ integer periods, meaning
\[\oint_\Sigma F \in 2\pi \bZ\]
over all $p+1$-submanifolds $\Sigma$, there is a $U(1)$ $p$-form gauge field $B$ which extends $F$ by $h,c$ as in theorem 1. Moreover, any two such gauge fields differ by a flat gauge field.
\end{prop}
\begin{proof}
Since $F$ has $2\pi$ integer periods, we can find a $c$ with periods $F/2\pi$. This defines $c$ up to torsion pieces and $h$ up to $h'$ with $dh' = 0$. These are precisely the flat gauge fields.
\end{proof}

\end{document}